\documentclass{article}

\usepackage[total={8.5in,11in},top=1in,left=1in,right=1in,bottom=1in]{geometry}
\usepackage{mathrsfs}
\usepackage{amssymb,amsfonts,amsmath,amsthm}
\usepackage{graphicx}
\usepackage{subcaption} 
\usepackage[noend]{algpseudocode}
\usepackage{algorithmicx,algorithm}
\usepackage{lineno}
\usepackage{verbatim}
 
\usepackage[textsize=footnotesize]{todonotes}
\usepackage{tabularx}
\usepackage{graphicx}
\usepackage{enumerate}
\usepackage{caption}
\newtheorem{theorem}{Theorem}[section]
\newtheorem{lemma}[theorem]{Lemma}

\newtheorem{observation}[theorem]{Observation}

\newcommand{\bigoh}{\mathcal{O}}
\usepackage{microtype}
\newif\iflong
\newif\ifshort

\longtrue
\iflong
\else
\shorttrue
\fi

\usepackage{tikz}
\usetikzlibrary{calc}
\usetikzlibrary{decorations.shapes}
\tikzset{
    ncbar angle/.initial=90,
    ncbar/.style={
        to path=(\tikztostart)
        -- ($(\tikztostart)!#1!\pgfkeysvalueof{/tikz/ncbar angle}:(\tikztotarget)$)
        -- ($(\tikztotarget)!($(\tikztostart)!#1!\pgfkeysvalueof{/tikz/ncbar angle}:(\tikztotarget)$)!\pgfkeysvalueof{/tikz/ncbar angle}:(\tikztostart)$)
        -- (\tikztotarget)
    },
    ncbar/.default=0.5cm,
}

\tikzset{square left brace/.style={ncbar=0.2cm}}
\tikzset{square right brace/.style={ncbar=-0.2cm}}

\tikzset{round left paren/.style={ncbar=0.5cm,out=120,in=-120}}
\tikzset{round right paren/.style={ncbar=0.5cm,out=60,in=-60}}

\tikzset{decorate sep/.style 2 args={decorate,decoration={shape backgrounds,shape=circle,shape size=#1,shape sep=#2}}}

\usepackage{tkz-euclide}
\def\boxit#1{\vbox{\hrule\hbox{\vrule\kern4pt
  \vbox{\kern1pt#1\kern1pt}
\kern2pt\vrule}\hrule}}

\newcommand{\nat}{\mathbb{N}}
\newcommand{\umatch}{p-{\sc Matching}}

\newcommand{\wmatcheq}{p-{\sc WT-Matching}}
\newcommand{\SSS}{\mathcal{S}}
\newcommand{\TT}{{\mathcal Trim}}
\newcommand{\RR}{{\mathcal Reduce}}
\newcommand{\CC}{{\mathcal Compact}}
\newcommand{\RC}{{\mathcal Red}$-${\mathcal Com}}
\newcommand\Algphase[1]{%
\vspace*{-.7\baselineskip}\Statex\hspace*{\dimexpr-\algorithmicindent-2pt\relax}\rule{\textwidth}{0.4pt}%
\Statex\hspace*{-\algorithmicindent}\textbf{#1}%
\vspace*{-.7\baselineskip}\Statex\hspace*{\dimexpr-\algorithmicindent-2pt\relax}\rule{\textwidth}{0.4pt}%
}
\newcommand{\NP}{\mbox{{\sf NP}}}

\newcommand{\paraproblem}[4]{\noindent {\sc #1}
\\
{\bf Given:} #2\\
{\bf Parameter:} #3\\
{\bf Goal:} #4}
\begin{document}

\title{\Large Optimal Streaming Algorithms for Graph Matching}

\author{
Jianer Chen\thanks{Department of Computer Science and Engineering, Texas A\&M University, College Station,  TX 77843, USA. Email: {\tt chen@cse.tamu.edu.}} \and 
Qin Huang\thanks{Department of Computer Science and Engineering,
    Texas A\&M University, College Station,  TX 77843, USA. Email: {\tt huangqin@tamu.edu.}}
    
\and Iyad Kanj\thanks{School of Computing, DePaul University, Chicago, IL 60604, USA. Email: {\tt ikanj@cs.depaul.edu}.} 
  
    \and Ge Xia\thanks{Department of Computer Science, Lafayette College, Easton, PA 18042, USA. Email: {\tt xiag@lafayette.edu.}} 
    }

\date{}

\maketitle

%\fancyfoot[R]{\scriptsize{Copyright \textcopyright\ 20XX by SIAM\\
%Unauthorized reproduction of this article is prohibited}}

% Depending on which copyright you agree to when you sign the copyright form, the copyright 
% can be changed to one of the following after commenting out the default copyright statement
% above.

%\fancyfoot[R]{\scriptsize{Copyright \textcopyright\ 20XX\\
%Copyright for this paper is retained by authors}}

%\fancyfoot[R]{\scriptsize{Copyright \textcopyright\ 20XX\\
%Copyright retained by principal author's organization}}

%\pagenumbering{arabic}
%\setcounter{page}{1}%Leave this line commented out.
\thispagestyle{empty} 
\begin{abstract}  
We present parameterized streaming algorithms for the graph matching problem in both the dynamic and the insert-only models. For the dynamic streaming model, we present a one-pass algorithm that, \emph{w.h.p.}~(with high probability), computes a maximum-weight $k$-matching of a weighted graph in $\tilde{O}(Wk^2)$ space\footnote{The notation $\tilde{O}()$ hides a poly-logarithmic factor in the input size.} and that has $\tilde{O}(1)$ update time, where $W$ is the number of distinct edge weights. For the insert-only streaming model, we present a one-pass algorithm that runs in $\bigoh(k^2)$ space and has $\bigoh(1)$ update time, and that, \emph{w.h.p.}, computes a maximum-weight $k$-matching of a weighted graph. The space complexity and the update-time complexity achieved by our algorithms for unweighted $k$-matching in the dynamic model and for weighted $k$-matching in the insert-only model are optimal.

A notable contribution of this paper is that the presented algorithms {\it do not} rely on the apriori knowledge/promise that the cardinality of \emph{every} maximum-weight matching of the input graph is upper bounded by the parameter $k$. This promise has been a critical condition in previous works, and lifting it required the development of new tools and techniques. 
% \begin{abstract}  
% We present parameterized streaming algorithms for the graph matching problem in both 
% the dynamic and the insert-only streaming models. Our results improve and generalize 
% several previous results.  
%
% For the dynamic streaming model, among our results is a one-pass sketch-based algorithm 
% that, for a weighted graph $G$, \emph{w.h.p.}~(with high probability) computes a maximum-weight 
% $k$-matching of $G$ (if $G$ contains a $k$-matching). The algorithm runs in $\tilde{O}(Wk^2)$ 
% space\footnote{The notation $\tilde{O}()$ hides a poly-logarithmic factor in the input size.} 
% and has $\tilde{O}(1)$ update time. A consequence of the aforementioned result is a one-pass 
% randomized streaming algorithm for $k$-matching in unweighted graphs that runs in $\tilde{O}(k^2)$ 
% space and has $\tilde{O}(1)$ update time.  
%
% For the insert-only model, we present a one-pass algorithm for weighted $k$-matching that runs 
% in $\bigoh(k^2)$ space and has $\bigoh(1)$ update time, and that \emph{w.h.p.}~computes a 
% maximum-weight $k$-matching in the graph. 
% 
% For unweighted $k$-matching in the dynamic model and weighted $k$-matching in the insert-only 
% model, the space complexity and the update-time complexity achieved by our algorithms are optimal. 
\end{abstract}
\section{Introduction}
%\subsection{Problem Definition and Motivation}
\label{subsec:def}

Emerging applications in big-data involve processing graphs of tremendous size~\cite{stream_survey}. 
For such applications, it is infeasible to store the graph when processing it.  This issue has given rise to a new computational model, referred to as the \emph{graph streaming model}. A graph \emph{stream} $\SSS$ for an underlying graph $G$ is a sequence of elements of the form $(e, op)$, where $op$ is an operation performed to edge $e$. % $\in E(G)$. 
In the \emph{insert-only} streaming model, each operation is an edge-insertion, while in the \emph{dynamic} streaming model each operation is either an edge-insertion or an edge-deletion (with a specified weight if $G$ is weighted). The graph streaming model demands performing the computation within limited space and time resources.

The graph matching problem, both in unweighted and weighted graphs, is one of the most extensively-studied problems in the streaming model. There has been a vast amount of work on its approximation and parameterized complexity, as will be discussed shortly.   

A matching $M$ in a graph $G$ is a \emph{$k$-matching} if $|M|=k$. A \emph{maximum-weight $k$-matching} in a weighted graph $G$ is a $k$-matching whose weight is maximum over all $k$-matchings in $G$. In this paper, we study parameterized streaming algorithms for the weighted and unweighted $k$-matching problem in both the dynamic and the insert-only streaming models. In these problems, we are given a graph stream and a parameter $k \in \nat$, and the goal is to compute a $k$-matching or a maximum-weight $k$-matching. We present results that improve several results in various aspects and that achieve optimal complexity upper bounds.

\subsection{Related Work}
\label{subsec:related}
    Most of the previous works on the graph matching problem in the streaming model have focused on approximating a maximum matching~\cite{soda16-5, soda17, soda16-9, soda16-15, soda16-18, soda16-23, soda16-26, soda16-31, soda16-32, soda16-37, soda16-38, soda-21, soda16-41, new-soda-17, soda16-48}, with the majority of these works pertaining to the (simpler) insert-only  model. The most relevant to ours are the works of~\cite{Chitnis2016, spaa2015, soda15, soda16-22}, which studied parameterized streaming algorithms for the maximum matching problem. We survey these works next.

Under the promise that the cardinality of \emph{every} maximal matching at every instant of the stream is at most $k$, the authors in~\cite{spaa2015,soda15} presented a one-pass dynamic streaming algorithm that \emph{w.h.p.}~computes a maximal matching in an unweighted graph stream. Their algorithm runs in $\tilde{O}(k^2)$ space and has $\tilde{O}(k^2)$ update time.  

The authors in~\cite{Chitnis2016} considered the problem of computing a maximum matching in the dynamic streaming model. For an unweighted graph $G$, under the promise that the cardinality of a maximum matching at every instant of the stream is at most $k$, they presented a sketch-based algorithm that \emph{w.h.p.}~computes a maximum matching of $G$, and that runs in $\tilde{O}(k^2)$ space and has $\tilde{O}(1)$ update time. They proved an $\Omega(k^2)$ lower bound on the space complexity of any randomized algorithm for the parameterized maximum matching problem, even in the insert-only model, thus showing that the space complexity of their algorithm is optimal (modulo a poly-logarithmic factor); their lower bound result carries over to the $k$-matching problem.
%the problem of computing a $k$-matching.  
They extended their algorithm to weighted graphs, and presented under the same promise an algorithm for computing a maximum-weight matching that runs in space $\tilde{O}(k^2 W)$ and has $\tilde{O}(1)$ update time, where $W$ is the number of distinct edge weights. 

We remark that the previous work on the weighted matching problem in the streaming model~\cite{Chitnis2016}, as well as our current work, assumes that the weight of each edge remains the same during the stream. Other works on weighted graph streams make the same assumption~\cite{ahn2012graph,ahn2015correlation,goel2012single,kapralov2014spanners}. The reason behind this assumption is that---as shown in this paper, if this assumption is lifted, we can derive a lower bound on the space complexity of the $k$-matching problem that is at least linear in the size of the graph, and hence, can be much larger than the desirable space complexity. 

Fafianie and Kratsch~\cite{soda16-22} studied kernelization streaming algorithms in the insert-only model for the \NP-hard $d$-{\sc Set Matching} problem (among others), which for $d=2$, is equivalent to the $k$-matching problem in unweighted graphs. Their result implies a one-pass kernelization streaming algorithm for $k$-matching in unweighted graphs that computes a kernel of size $\bigoh(k^2 \log k)$, and that runs in $\bigoh(k^2)$ space and has $\bigoh(\log k)$ update time.
 
We mention that Chen et al.~\cite{chen2020} studied algorithms for $k$-matching in unweighted and weighted graphs in the RAM model; their algorithms use limited computational resources and \emph{w.h.p.}~compute a $k$-matching.   Clearly, the RAM model is very different from the streaming model. %For instance, their algorithm assumes %that the graph is stored as an adjacency list, whereas in the streaming %model the graph edges are inserted in an arbitrary order. 
In order to translate their algorithm to the streaming model, it would require $\Omega($n$ k)$ space and multiple passes, where $n$ is the number of vertices. However, we mention that (for the insert-only model), in one of the steps of our algorithm, we were inspired by a graph operation for constructing a reduced graph, which was introduced in their paper. 

Finally, there has been some work on computing matchings in special graph classes, and with respect to parameters other than the cardinality of the matching (e.g., see~\cite{Niedermeier4,  Niedermeier2, Niedermeier3, Niedermeier1}).

\iflong As is commonly the case in the relevant literature, we work under the assumption that each basic operation on words takes constant time and space. \fi

\subsection{Results and Techniques}
\noindent {\bf Results and Techniques for the Dynamic Model.} We give a one-pass sketch-based one-sided error streaming algorithm that, for a weighted graph $G$, if $G$ contains a $k$-matching then, with probability at least 
 $1-\frac{11}{20k^3\ln (2k)}$, the algorithm computes a maximum-weight $k$-matching of $G$ and if $G$ does not contain a $k$-matching then the algorithm reports that correctly. The algorithm runs in $\tilde{O}(Wk^2)$ space and has $\tilde{O}(1)$ update time, where $W$ is the number of distinct weights in the graph. A byproduct of this result is a one-pass one-sided error streaming algorithm for unweighted $k$-matching running in $\tilde{O}(k^2)$ space and having $\tilde{O}(1)$ update time. For $k$-matching in unweighted graphs, the space and update-time complexity of our algorithm are optimal (modulo a poly-logarithmic factor of $k$). 

%\todo[inline]{The second point?}
The above results achieve the same space and update-time complexity as the results in~\cite{Chitnis2016}, but generalize them in the following ways. First, our algorithm can be used to solve the unweighted/weighted $k$-matching problem for any $k \in \nat$, whereas the algorithm in~\cite{Chitnis2016} can \emph{only} be used to solve the maximum (resp.~maximum-weight) matching problem \emph{and} under the promise that the given parameter $k$ is at least as large as the cardinality of every maximum (resp.~maximum-weight) matching. In particular, if one wishes to compute a maximum-weight $k$-matching where $k$ is smaller than the cardinality of a maximum-weight matching, then the algorithm in~\cite{Chitnis2016} cannot be used or, in the unweighted case, incurs a complexity that depends on the cardinality of the maximum matching. Second, the correctness of the algorithm in~\cite{Chitnis2016} relies heavily on the promise that the cardinality of every maximum (resp.~maximum-weight) matching is at most $k$, in the sense that, if this promise is not kept then, for certain instances of the problem, \emph{w.h.p.}~the subgraph returned by the algorithm in~\cite{Chitnis2016} does not contain a maximum (resp.~maximum-weight) matching.  Third, if we are to work under the same promise as in~\cite{Chitnis2016}, which is that the cardinality of every maximum (resp.~maximum-weight) matching is at most $k$, then  \emph{w.h.p.}, our algorithm ~computes a maximum (resp.~maximum-weight) matching. Therefore, in these respects, our results generalize those in~\cite{Chitnis2016}. 

Another byproduct of our result for weighted $k$-matching is a one-pass streaming approximation algorithm that, for any $\epsilon > 0$, \emph{w.h.p.}~computes a $k$-matching that is within a factor of $1+\epsilon$ from a maximum-weight $k$-matching in $G$; the algorithm runs in $\tilde{O}(k^2\epsilon^{-1}\log W')$ space and has $\tilde{O}(1)$ update time, where $W'$ is the ratio of the maximum edge-weight to the minimum edge-weight in $G$. This result matches the approximation result in~\cite{Chitnis2016}, which achieves the same bounds, albeit under the aforementioned promise. 

We complement the above with a space lower-bound result \ifshort {\rm ($\spadesuit$)} \fi showing that, if the restriction that the weight of each edge remains the same during the stream is lifted, which---as mentioned before--is a standard assumption, then \emph{even} computing a 1-matching whose weight is within a $(6/5)$-factor from the maximum-weight 1-matching, by any randomized streaming algorithms requires space that is at least linear in the size of the graph, and hence, can be much larger than the desirable space complexity.   \ifshort This result is given in the appendix.\fi

As mentioned before, the results in~\cite{Chitnis2016} rely on the assumption that the cardinality of every maximum (resp.~maximum-weight) matching is at most $k$. This assumption is essential for their techniques to work since it is used to upper bound the number of ``large'' vertices of degree at least $10k$ by $\bigoh(k)$, and the number of ``small'' edges whose both endpoints have degree at most $10k$ by $\bigoh(k^2)$. These bounds allow the sampling of a set of edges that \emph{w.h.p.}~contains a maximum (resp.~maximum-weight) matching.

To remove the reliance on the promise, we prove a structural result that can be useful in its own right for $k$-subset problems (in which the goal is to compute a $k$-subset with certain prescribed properties from some universe $U$).  Intuitively, the result states that, for any $k$-subset $S \subseteq U$, \emph{w.h.p.}~we can compute $k$ subsets $T_1, \ldots, T_k$ of $U$ that interact ``nicely'' with $S$. More specifically, (1) the sets $T_i$, for $i \in [k]$, are pairwise disjoint, (2) $S$ is contained in their union $\bigcup_{i\in [d]} T_i$, and  (3) each $T_i$ contains exactly one element of $S$.  To prove the theorem, we show that we can randomly choose an $\bigoh(\log k)$-wise independent hash function that partitions $S$ ``evenly''. We then show that we can randomly choose $\bigoh(k/\log k)$-many hash functions, from a set of universal hash functions, such that there exist $k$ integers $p_1, \ldots, p_k$, satisfying that $T_i$ is the pre-image of $p_i$ under one of the chosen hash functions.  

We then apply the above result to obtain the sets $T_i$ of vertices that \emph{w.h.p.}~induce the edges of the desired $k$-matching. Afterwards, we use $\ell_0$-sampling to select a smaller subset of edges induced by the vertices of the $T_i$'s that \emph{w.h.p.}~contains the desired $k$-matching. From this smaller subset of edges, a maximum-weight $k$-matching can be extracted.

\noindent {\bf Results and Techniques for the Insert-Only Model.} We present a one-pass one-sided error algorithm for computing a maximum-weight $k$-matching that runs in $\bigoh(k^2)$ space and has $\bigoh(1)$ update time.  The space and update-time complexity of our algorithm are optimal.

Our techniques rely on partitioning the graph (using hashing), and defining an auxiliary graph whose vertices are the different parts of the partition; the auxiliary graph is updated during the stream. By querying  this auxiliary graph, the algorithm can compute a ``compact'' subgraph of size  $\bigoh(k^2)$ that, \emph{w.h.p.}, contains the edges of the desired $k$-matching. A maximum-weight $k$-matching can then be extracted from this compact subgraph.

Fafianie and Kratsch~\cite{soda16-22} studied kernelization streaming algorithms in the insert-only model for the \NP-hard $d$-{\sc Set Matching} problem, which for $d=2$, is equivalent to $k$-matching in unweighted graphs. Their result implies a one-pass kernelization streaming algorithm for $k$-matching that computes a kernel of size $\bigoh(k^2 \log k)$ bits using $\bigoh(k^2)$ space and $\bigoh(\log k)$ update time. In comparison, our algorithm computes a compact subgraph, which is a kernel of the same size as in~\cite{soda16-22} (and from which a $k$-matching can be extracted); moreover, our algorithm treats the more general weighted case, and achieves a better update time of $\bigoh(1)$ than that of~\cite{soda16-22}, while matching their upper bound on the space complexity.

\ifshort
\section{Preliminaries}
\label{sec:prelim}
We refer to the appendix and to the following books for more details~\cite{Cygan,Diestel,fptbook,flumgrohe,upfal,niedermeier}. We write ``\emph{u.a.r.}''~as an abbreviation for ``uniformly at random''. For a positive integer $i$, let $[i]^-$ denote the set of numbers $\{0, 1, \ldots, i-1\}$, $[i]$ denote the set of numbers $\{1, \ldots, i\}$, and $\llcorner i \lrcorner$ denote the binary representation of $i$.

\noindent {\bf Computational Model \& Problem Definition.}
In a parameterized graph streaming problem $Q$, we are given an instance of the form $(\SSS, k)$, where $\SSS$ is graph stream of some underlying graph $G$ and $k \in \nat$, and we are asked to compute a solution for $(\SSS, k)$~\cite{soda15}. 

A parameterized streaming algorithm ${\cal A}$ for $Q$ generally uses a \emph{sketch}, which is a data structure that supports a set of update operations~\cite{alon, soda15, siam03, jacm06}. The algorithm ${\cal A}$ can update the sketch after processing each element of $\SSS$; the time taken to update the sketch---after processing an element---is referred to as the \emph{update time} of the algorithm. The space used by ${\cal A}$ is the space needed to compute and store the sketch, and that needed to solve the instance of the problem.  
to extract a solution from the sketch.  We study the following problems: 
\begin{itemize}
    \item \umatch: Given a graph stream $\SSS$ of an unweighted graph $G$ and a parameter $k$, compute a $k$-matching in $G$ or report that no $k$-matching exists.
    \item \wmatcheq: A graph stream $\SSS$ of a weighted graph $G$ and a parameter $k$, compute a $k$-matching of maximum weight in $G$ or report that no $k$-matching exists.
\end{itemize}

We will assume that $V(G)=\{0, \ldots, n-1\}$, and that the length of $\SSS$ is polynomial in $n$. We will design parameterized streaming algorithms for the above problems. Our algorithms first extract a subgraph $G'$ of the graph stream $G$ such that \emph{w.h.p.}~$G'$ contains a $k$-matching or a maximum-weight $k$-matching 
of $G$ if and only if $G$ contains one. In the case where the size of $G'$ is a function of $k$, such algorithms are referred to as \emph{kernelization streaming algorithms}~\cite{Chitnis2016}.  We note that result in~\cite{Chitnis2016} also computes a subgraph containing the edges of the maximum (resp.~maximum-weight) matching, without computing the matching itself, as there are efficient algorithms for extracting a maximum matching (resp.~maximum-weight) or a $k$-matching (resp.~maximum-weight $k$-matching) from that subgraph~\cite{gabow1,gabow2,tarjanbook}.

\noindent {\bf Probability.} A set $\{X_1, \ldots, X_j\}$ of random variables is \emph {$\lambda$-wise independent} if for any subset 
$J \subseteq \{1, \ldots, j\}$ with $|J| \le \lambda$ and for any values $x_i$: 
$\Pr(\wedge_{i\in J} X_i=x_i)=\prod_{i\in J} \Pr(X_i=x_i)$. 
 
\begin{theorem} {\rm (Theorem 2 in \cite{Chernoff1995})} \label{chernoff-bound}
Given any 0-1 random variables $X_1, \ldots, X_j$, let $X=\sum_{i=1}^j X_i$ and $\mu=E[X]$. For any $\delta > 0$, 
if the $X_i$'s are $\lceil \mu\delta \rceil$-wise independent, then:\\ 
\hspace*{1cm}$ \Pr(X \ge \mu(1+\delta)) \le 
  \begin{cases}
    e^{-\mu \delta^2/3}       & \quad \text{if } \delta <1 \\
    e^{-\mu \delta/3}     & \quad \text{if } \delta \ge 1
  \end{cases}
$
\end{theorem}

\noindent {\bf $\ell_0$-Samplers.} 
Let $0<\delta<1$ be a parameter. Let 
$\SSS=(i_1, \Delta_1), \ldots, (i_p, \Delta_p), \ldots$ 
be a stream of updates of an underlying vector $\textbf{x} \in \mathbb{R}^n$, where $i_j \in [n]$ and $\Delta_j \in \mathbb{R}$.
The $j$-th update $(i_j, \Delta_j)$ updates the $i_j$-th 
coordinate of $\textbf{x}$ by setting $\textbf{x}_{i_j}=\textbf{x}_{i_j}+\Delta_j$. An 
{\it $\ell_0$-sampler} for $\textbf{x} \neq 0$ either fails with probability at most $\delta$, or conditioned on not failing, for any non-zero coordinate $\textbf{x}_j$ of $\textbf{x}$, returns the pair $(j, \textbf{x}_j)$ with probability 
$\frac{1}{||\textbf{x}||_0}$, where $||\textbf{x}||_0$ is the $\ell_0$-norm of $\textbf{x}$, which is the same as the number of non-zero 
coordinates of $\textbf{x}$. (We refer to \cite{l0-sampling}.)

\begin{lemma}[Follows from Theorem~2.1 in~\cite{Chitnis2016}]
\label{lem:sampler}
 Let $0 < \delta < 1$ be a parameter. There exists a linear sketch-based $\ell_0$-sampler algorithm that, given a dynamic graph stream, either returns FAIL with probability at most $\delta$, or returns an edge
chosen \emph{u.a.r.}~from the edges of the stream that have
been inserted and not deleted. This algorithm
can be implemented using $\bigoh(\log^2{n} \cdot \log(\delta^{-1}
))$ bits of space and $\tilde{O}(1)$ update time, where $n$ is the number of vertices.
\end{lemma}

\noindent {\bf Hash Functions.}
Let $U$ be a universe of elements that we will refer to as \emph{keys}.  
For a collection $\mathcal{H}$ of hash functions and a hash function $h \in \mathcal{H}$, we write 
$h \in_{u.a.r} \mathcal{H}$ to denote that $h$ is chosen \emph{u.a.r.}~from $\mathcal{H}$. Let $S \subseteq U$ and $r$ be a positive integer. A 
hash function $h: U \longrightarrow [r]^-$ is \emph{perfect} w.r.t.~$S$ if it is injective on $S$. 

A set $\mathcal{H}$ of hash functions, each mapping $U$ to $[r]^-$, is 
called \emph{universal} if for each pair of distinct keys $x, y \in U$, the number of hash functions $h\in \mathcal{H}$ for which $h(x)=h(y)$ is at most $|\mathcal{H}|/r$. 
Let $p \geq |U|$ be a prime number.  A universal set $\mathcal{H}$ of hash functions from $U$ to $[r]^-$ can be constructed as follows 
(see chapter 11 in \cite{Cormen}):
$\mathcal{H}=\{h_{a,b,r} \mid 1\le a \le p-1, 0\le b \le p-1\},$
where $h_{a,b,r}$ is defined as $h_{a,b,m}(x)=((ax+b) \mod p) \mod r$.

\begin{theorem} [Theorem 11.9 in \cite{Cormen}] \label{lemma-hash}
Let $U$ be a universe and $\mathcal{H}$ be a universal set of hash functions, each mapping $U$ to
$[r^2]^-$. For any set $S \subseteq U$ of $r$ elements and any hash function $h \in_{u.a.r.} \mathcal{H}$, the probability that $h$ is perfect w.r.t.~$S$ is larger than $1/2$.
\end{theorem}

A family $\mathcal{H}$ of hash functions, each mapping $U$ to $[r]^-$, is called $\kappa$-wise independent if for any $\kappa$ distinct keys $x_1, x_2, ..., x_{\kappa} \in U$, 
and any $\kappa$  (not necessarily 
distinct) values $a_1, a_2, ..., a_{\kappa} \in [r]^-$, we have 
$\Pr_{h \in_{u.a.r} \mathcal{H}} [ h(x_1)=a_1 \wedge h(x_2)=a_2 \wedge \cdots \wedge h(x_{\kappa})= a_{\kappa} ] = \frac{1}{r^{\kappa}}.$
 
 Let $\mathbb{F}$ be a finite field. A $\kappa$-wise independent family $\mathcal{H}$ of hash functions can be constructed as follows (See Construction 3.32 in \cite{Salil}):
 $\mathcal{H} = \{ h_{a_0, a_1, \ldots, a_{\kappa-1}}: \mathbb{F} \rightarrow \mathbb{F}\},$
 where $h_{a_0, a_1, \ldots, a_{\kappa-1}}(x) = a_0+ a_1x + \cdots + a_{\kappa-1}x^{\kappa-1}$ for 
 $a_0, \ldots, a_{\kappa-1} \in \mathbb{F}$.

\begin{theorem}[Corollary 3.34 in \cite{Salil}]  \label{evaluate-time}
For every $u, d, \kappa \in \mathbb{N}$, there is a family of $\kappa$-wise independent functions 
$\mathcal{H}= \{ h: \{0,1\}^ u \rightarrow \{0, 1\}^{d}\}$ such that choosing a random 
function from $\mathcal{H}$ takes space $\bigoh(\kappa \cdot (u+d))$. 
Moreover, evaluating a function from $\mathcal{H}$ takes time polynomial in $u, d, \kappa$. 
\end{theorem}

\fi

\iflong
\section{Preliminaries}
\label{sec:prelim}
For a positive integer $i$, let $[i]^-$ denote the set of numbers $\{0, 1, \ldots, i-1\}$, $[i]$ denote the set of numbers $\{1, \ldots, i\}$, and $\llcorner i \lrcorner$ denote the binary representation of $i$. We write ``u.a.r.'' as an abbreviation for ``uniformly at random''. 

\subsection{Parameterized Complexity} 
\label{subsec:pc}
A \emph{parameterized problem} $Q$ is a subset of 
$\Sigma^* \times \mathbb{N}$, where $\Sigma$ is a fixed, finite alphabet. Each instance is a pair 
$(I, k)$, where $k\in \mathbb{N}$ is called the parameter. A parameterized problem $Q$ is \emph{kernelizable} if there exists a polynomial-time reduction that maps an instance $(I, k)$ of $Q$ to another instance 
$(I', k')$ such that (1) $k' \le g(k)$ and $|I'| \le g(k)$, where $g$ is a computable function and 
$|I'|$ is the length of the instance, and $(2)$ $(I, k)$ is a yes-instance of $Q$ if and only if 
$(I', k')$ is a yes-instance of $Q$. The polynomial-time reduction is called the \emph{kernelization algorithm} and the instance $(I', k')$ is called the kernel of $(I,k)$. We refer to~\cite{Cygan} for more information. 

\subsection{Graphs and Matching} All graphs discussed in this paper are undirected and simple. 
Let $G$ be a graph. We write $V(G)$ and $E(G)$ for the vertex-set and edge-set of $G$, respectively, and write $uv$ for the edge whose endpoints are $u$ and $v$.  
A \emph{matching} $M \subseteq E(G)$ is a set of edges such that no two distinct edges in $M$ share the same endpoint. A matching $M$ is a 
\emph{$k$-matching} if $|M|=k$. A weighted graph $G$ is a graph associated with a weight function $wt: E(G) \longrightarrow \mathbb{R}$; we denote the weight of an edge $e$ by $wt(e)$. Let $M$ be a matching in a weighted graph $G$. The weight of $M$, $wt(M)$, is the sum of the weights of the edges in $M$, that is, $wt(M)=\sum_{e \in M} wt(e)$.   
A \emph{maximum-weight $k$-matching} in a weighted graph $G$ is a $k$-matching whose weight is maximum over all $k$-matchings in $G$.

\subsection{The Graph Streaming Model} 
\label{subsec:graphstream}

A graph \emph{stream} $\SSS$ for an underlying graph $G=(V, E)$ is a sequence of elements, each of the form $(e, op)$, where $op$ is an update to edge $e \in E(G)$. Each update could be an \emph{insertion} of an edge, a \emph{deletion} of an edge, or in the case of a weighted graph an update to the weight of an edge in $G$ (and would include the weight of the edge in that case). In the \emph{insert-only graph streaming} model, a  graph $G=(V, E)$ is given as a stream $\SSS$ of elements in which each operation is an edge-insertion, while in the \emph{dynamic graph streaming} model a graph $G=(V, E)$ is given as a stream $\SSS$ of elements in which the operations could be either edge-insertions or edge-deletions (with specified weights in case $G$ is weighted).

We assume that the vertex set $V(G)$ contains $n$ vertices, identified with the integers $\{0, \ldots, n-1\}$ for convenience, and that the length of the stream $\SSS$ is polynomial in $n$. Therefore, we will treat $v\in V$ as a unique number $v \in [n]^-$. Without loss of generality, since the graph $G$ is undirected, we will assume that the edges of the graph are of the form $uv$, where $u < v$. Since $G$ can have 
at most ${ n \choose 2 } = n(n-1)/2$ edges, each edge can be represented as a unique number in $[n(n-1)/2]^-$. At the beginning of the stream $\SSS$, the stream corresponds to a graph with an empty edge-set. For weighted graphs, 
we assume that the weight of an edge is specified when the edge is inserted or deleted.

\subsection{Computational Model \& Problem Definition}
\label{subsec:problemdefinition}

In a parameterized graph streaming problem $Q$, we are given an instance of the form $(\SSS, k)$, where $\SSS$ is a graph stream of some underlying graph $G$ and $k \in \nat$, and we are queried for a solution for $(\SSS, k)$ either at the end of $\SSS$ or after some arbitrary element/operation in $\SSS$~\cite{soda15}. 

A parameterized streaming algorithm ${\cal A}$ for $Q$ generally uses a \emph{sketch}, which is a data structure that supports a set of update operations~\cite{alon, soda15, siam03, jacm06}. The algorithm ${\cal A}$ can update the sketch after reading each element of $\SSS$; the time taken to update the sketch---after reading an element---is referred to as the \emph{update time} of the algorithm. The space used by ${\cal A}$ is the space needed to compute and store the sketch, and that needed to answer a query on the instance (based on the sketch). The time complexity of the algorithm is the time taken 
to extract a solution from the sketch when answering a query.

We will consider parameterized problems in which we are given a graph stream $\SSS$ and $k \in \nat$, and the goal is to compute a $k$-matching or a maximum-weight $k$-matching in $G$ (if one exists). We will consider both the unweighted and weighted $k$-matching problems, referred to as \umatch{} and \wmatcheq{}, respectively, and in both the insert-only and the dynamic streaming models. For \wmatcheq{}, we will follow the standard literature assumption~\cite{Chitnis2016}, which is that the weight of every edge remains the same throughout $\SSS$; we will consider in Subsection~\ref{subsec:lowerbounds} a generalized version of this problem that allows the weight of an edge to change during the course of stream, and prove lower bound results for this generalization.  We formally define the problems under consideration:
 
\paraproblem{\umatch}{A graph stream $\SSS$ of an unweighted graph 
$G$}{$k$}{Compute a $k$-matching in $G$ or report that no $k$-matching exists}

The parameterized {\sc Weighted Graph Matching} (p-{\sc WT-Matching}) problem is defined as fellows:

\paraproblem{\wmatcheq}{A graph stream $\SSS$ of a weighted graph 
$G$}{$k$}{Compute a maximum-weight $k$-matching in $G$ or report that no $k$-matching exists}

Clearly, if $k>n/2$ then $G$ does not contain a $k$-matching. Therefore, we may assume henceforth that $k\le n/2$.

We will design parameterized streaming algorithms for the above problems. Our algorithms first extract a subgraph $G'$ of the graph stream $G$ such that \emph{w.h.p.}~$G'$ contains a $k$-matching or a maximum-weight $k$-matching 
of $G$ if and only if $G$ contains one. In the case where the size of $G'$ is a function of $k$, such algorithms are referred to as \emph{kernelization streaming algorithms}~\cite{Chitnis2016}.  We note that result in~\cite{Chitnis2016} also computes a subgraph containing the edges of the maximum (resp.~maximum-weight) matching, without computing the matching itself, as there are efficient algorithms for extracting a maximum matching (resp.~maximum-weight) or a $k$-matching (resp.~maximum-weight $k$-matching) from that subgraph~\cite{gabow1,gabow2,tarjanbook}.

\subsection{Probability} For any probabilistic events $E_1, E_2, \ldots, E_r$, the \emph{union bound} states 
that $\Pr(\bigcup_{i=1}^r E_i) \le \sum_{i=1}^r \Pr(E_i)$. For any random variables $X_1, \ldots, X_r$ whose
expectations are well-defined, the linearity of expectation states that $E[\sum_{i=1}^r X_i] = \sum_{i=1}^r E[X_i]$, 
where $E[X_i]$ is the expectation of $X_i$. 
A set of discrete random variables $\{X_1, \ldots, X_j\}$ is called \emph {$\lambda$-wise independent} if for any subset 
$J \subseteq \{1, \ldots, j\}$ with $|J| \le \lambda$ and for any values $x_i$, we have 
$\Pr(\wedge_{i\in J} X_i=x_i)=\prod_{i\in J} \Pr(X_i=x_i)$. 
A random variable is called a \emph{0-1 random variable }, if it only takes 
one of the two values 0, 1. The following theorem bounds the tail probability of the sum of 
0-1 random variables with limited independence (see Theorem 2 in \cite{Chernoff1995}):
\begin{theorem}  \label{chernoff-bound}
Given any 0-1 random variables $X_1, \ldots, X_j$, let $X=\sum_{i=1}^j X_i$ and $\mu=E[X]$. For any $\delta > 0$, 
if the $X_i$'s are $\lceil \mu\delta \rceil$-wise independent, then 
\[ \Pr(X \ge \mu(1+\delta)) \le 
  \begin{cases}
    e^{-\mu \delta^2/3}       & \quad \text{if } \delta <1 \\
    e^{-\mu \delta/3}     & \quad \text{if } \delta \ge 1
  \end{cases}
\]

\end{theorem}

\subsection{$\ell_0$-sampler} \label{subsec:sampler} Let $0<\delta<1$ be a parameter. Let 
$\SSS=(i_1, \Delta_1), \ldots, (i_p, \Delta_p), \ldots$ 
be a stream of updates of an underlying vector $\textbf{x} \in \mathbb{R}^n$, where $i_j \in [n]$ and $\Delta_j \in \mathbb{R}$.
The $j$-th update $(i_j, \Delta_j)$ updates the $i_j$-th 
coordinate of $\textbf{x}$ by setting $\textbf{x}_{i_j}=\textbf{x}_{i_j}+\Delta_j$. An 
{\it $\ell_0$-sampler} for $\textbf{x} \neq 0$ either fails with probability at most $\delta$, or conditioned on not failing, for any non-zero coordinate $\textbf{x}_j$ of $\textbf{x}$, returns the pair $(j, \textbf{x}_j)$ with probability 
$\frac{1}{||\textbf{x}||_0}$, where $||\textbf{x}||_0$ is the $\ell_0$-norm of $\textbf{x}$, which is the same as the number of non-zero 
coordinates of $\textbf{x}$. For more details, we refer to \cite{l0-sampling}.

Based on the results in~\cite{l0-sampling,pods11}, and as shown in~\cite{Chitnis2016}, we can develop a sketch-based 
$\ell_0$-sampler algorithm for a dynamic graph stream that samples an edge from the stream. More specifically, the following result was shown in~\cite{Chitnis2016}:

\begin{lemma}[Proof of Theorem~2.1 in~\cite{Chitnis2016}]
\label{lem:sampler}
 Let $0 < \delta < 1$ be a parameter. There exists a linear sketch-based $\ell_0$-sampler algorithm that, given a dynamic graph stream, either returns FAIL with probability at most $\delta$, or returns an edge
chosen \emph{u.a.r.}~amongst the edges of the stream that have
been inserted and not deleted. This $\ell_0$-sampler algorithm
can be implemented using $\bigoh(\log^2{n} \cdot \log(\delta^{-1}
))$ bits of space and $\tilde{O}(1)$ update time, where $n$ is the number of vertices of the graph stream.
\end{lemma}

\subsection{Hash Functions}
Let $U$ be a universe of elements that we will refer to as \emph{keys}. We can always identify the elements of $U$ with the numbers $0, \ldots, |U|-1$; therefore, without loss of generality, we will assume henceforth that $U=\{0, 1, \ldots, |U|-1\}$.
For a set $\mathcal{H}$ of hash functions  and a hash function $h \in \mathcal{H}$, we write 
$h \in_{u.a.r.} \mathcal{H}$ to denote that $h$ is chosen \emph{u.a.r.}~from $\mathcal{H}$. Let $S \subseteq U$ and $r$ be a positive integer. A 
hash function $h: U \longrightarrow [r]^-$ is \emph{perfect} w.r.t.~$S$ if it is injective on $S$ (i.e., no two distinct elements $x, y \in S$ cause a collision). 

A set $\mathcal{H}$ of hash functions, each mapping $U$ to $[r]^-$, is 
called \emph{universal} if for each pair of distinct keys $x, y \in U$, the number of hash functions $h\in \mathcal{H}$ for which $h(x)=h(y)$ is at most $|\mathcal{H}|/r$, or equivalently: 
$$ \Pr_{h \in_{u.a.r.} \mathcal{H}}[h(x)=h(y)] \le \frac{1}{r}. $$
Let $p \geq |U|$ be a prime number.  A universal set of hash functions $\mathcal{H}$ from $U$ to $[r]^-$ can be constructed as follows 
(see chapter 11 in \cite{Cormen}):

$$\mathcal{H}=\{h_{a,b,r} \mid 1\le a \le p-1, 0\le b \le p-1\},$$
where $h_{a,b,r}$ is defined as $h_{a,b,r}(x)=((ax+b) \mod p) \mod r$.

\begin{theorem} [Theorem 11.9 in \cite{Cormen}] \label{lemma-hash}
Let $U$ be a universe and $\mathcal{H}$ be a universal set of hash functions, each mapping $U$ to
$[r^2]^-$. For any set $S$ of $r$ elements in $U$ and any hash function $h \in_{u.a.r.} \mathcal{H}$, the probability that $h$ is perfect w.r.t.~$S$ is larger than $1/2$.
\end{theorem}

A set $\mathcal{H}$ of hash functions, each mapping $U$ to $[r]^-$, is called $\kappa$-wise independent if for any $\kappa$ distinct keys $x_1, x_2, ..., x_{\kappa} \in U$, 
and any $\kappa$  (not necessarily 
distinct) values $a_1, a_2, ..., a_{\kappa} \in [r]^-$, we have 
$$ \Pr_{h \in_{u.a.r.} \mathcal{H}} [ h(x_1)=a_1 \wedge h(x_2)=a_2 \wedge \cdots \wedge h(x_{\kappa})= a_{\kappa} ] = \frac{1}{r^{\kappa}}.$$
 
 Let $\mathbb{F}$ be a finite field. A $\kappa$-wise independent set $\mathcal{H}$ of hash functions can be constructed as follows (See Construction 3.32 in \cite{Salil}):
 $$ \mathcal{H} = \{ h_{a_0, a_1, \ldots, a_{\kappa-1}}: \mathbb{F} \rightarrow \mathbb{F}\},$$
 where $h_{a_0, a_1, \ldots, a_{\kappa-1}}(x) = a_0+ a_1x + \cdots + a_{\kappa-1}x^{\kappa-1}$ for 
 $a_0, \ldots, a_{\kappa-1} \in \mathbb{F}$.

\begin{theorem}[Corollary 3.34 in \cite{Salil}]  \label{evaluate-time}
For every $u, d, \kappa \in \mathbb{N}$, there is a family of $\kappa$-wise independent functions 
$\mathcal{H}= \{ h: \{0,1\}^ u \rightarrow \{0, 1\}^{d}\}$ such that choosing a random 
function from $\mathcal{H}$ takes space $\bigoh(\kappa \cdot (u+d))$. 
Moreover, evaluating a function from $\mathcal{H}$ takes time polynomial in $u, d, \kappa$. 
\end{theorem}
\fi
 
\section{The Toolkit}
\label{sec:structural}
In this section, we prove a theorem that can be useful in its own right for subset problems, that is, problems in which the goal is to compute a $k$-subset $S$ ($k \in \nat$) of some universe $U$ such that $S$ satisfies certain prescribed properties.  Intuitively, the theorem states that, for any $k$-subset $S \subseteq U$, \emph{w.h.p.}~we can compute $k$ subsets $T_1, \ldots, T_k$ of $U$ that interact ``nicely'' with $S$. More specifically, (1) the sets $T_i$, for $i \in [k]$, are pairwise disjoint,  (2) $S$ is contained in their union $\bigcup_{i\in [d]} T_i$, and  (3) each $T_i$ contains exactly one element of $S$.   

The above theorem will be used in Section~\ref{sec:dynamic} to design algorithms for p-{\sc Matching} and p-{\sc WT-Matching} in the dynamic streaming model. Intuitively speaking, the theorem will be invoked to obtain the sets $T_i$ of vertices that \emph{w.h.p.}~induce the edges of the desired $k$-matching; however, these sets may not necessarily constitute the desired subgraph as they may not have ``small'' cardinalities. Sampling techniques will be used to select a smaller set of edges induced by the vertices of the $T_i$'s that \emph{w.h.p.}~contains the edges of the $k$-matching.

To prove this theorem, we proceed in two phases. We give an intuitive description of these two phases next. \iflong We refer to Figure~\ref{fig:process} for illustration. \fi \ifshort (We refer to Figure~\ref{fig:process} in Section~\ref{sec:figures} in the appendix for illustration.)\fi 

%\todo[inline]{Qin: do we need to define $S_i=U_i \cap S$?}
In the first phase, we choose a hashing function $f$ u.a.r.~from an $\bigoh(\ln k)$-wise independent set of hash functions, which hashes $U$ to a set of $d_1 =\bigoh(k/\ln k)$ integers. We use $f$ to partition the universe $U$ into $d_1$-many subsets $U_i$, each consisting of all elements of $U$ that hash to the same value under $f$. Afterwards, we choose $d_1$ families $F_0, \ldots, F_{d_1-1}$ of hash functions, each containing $d_2= \bigoh(\ln k)$ functions, chosen independently and u.a.r.~from a universal set of hash functions. The family $F_i$, $i \in [d_1]^-$, will be used restrictively to map the elements of $U_i$.
 Since each family $F_i$ is chosen from a universal set of hash function, for the subset $S_i = S \cap U_i$, \emph{w.h.p.}~$F_i$ contains a hash function $f_i$ that is perfect w.r.t.~$S_i$; that is, under the function $f_i$ the elements of $S_i$ are distinguished. \iflong This concludes the first phase of the process, which is described in \textbf{Algorithm~\ref{good-hash}}.\fi 
 \ifshort This concludes the first phase of the process, which is described in \textbf{Algorithm~\ref{good-hash}}. \fi
 
\iflong
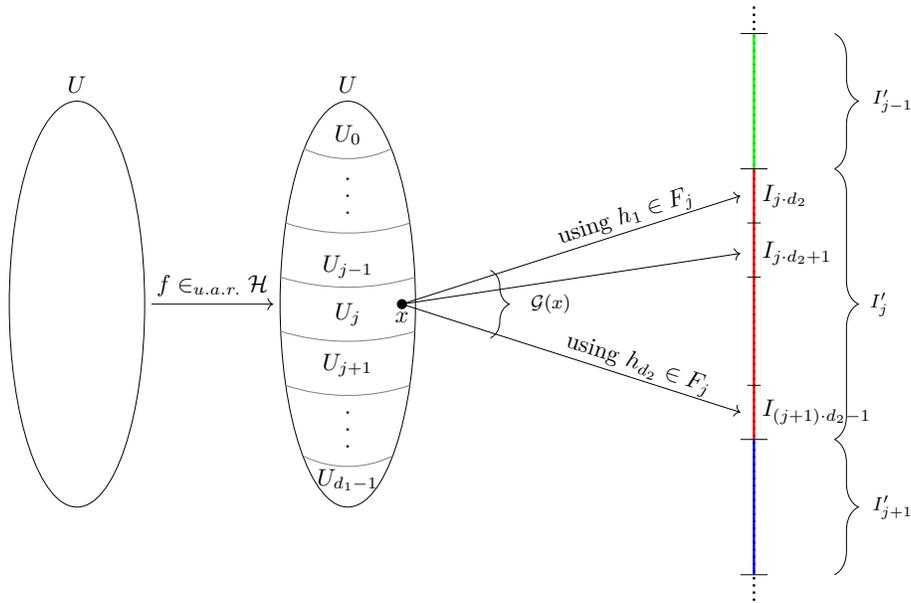
\begin{figure*}[ht]
\begin{center}
%\begin{adjustbox}{\textwidth}
\scalebox{0.9}{
\begin{tikzpicture}
\draw[shift={(0,3)}]   node[above]{$U$};
\draw (0,0) ellipse (1cm and 3cm);
\draw[shift={(4,3)}]   node[above]{$U$};
\draw (4,0) ellipse (1cm and 3cm);

     \draw[shift={(2,0)}]   node[above]{$f \in_{u.a.r.} {\mathcal H}$};
     \draw[->](1.1,0) -- (2.9,0);

\draw[shift={(4,2.5)}]   node[]{$U_0$};
\draw[shift={(4,1.6)}]   node[rotate=90]{. . .};
\draw[shift={(4,0.55)}]   node[]{$U_{j-1}$};
\draw[shift={(4,-0.15)}]   node[]{$U_{j}$};
\draw[shift={(4,-0.9)}]   node[]{$U_{j+1}$};
\draw[shift={(4,-1.85)}]   node[rotate=90]{. . .};
\draw[shift={(4,-2.6)}]   node[]{$U_{d_1-1}$};

\tkzDefPoint(3.35, 2.3){A}\tkzDefPoint(4, 2.15){B}\tkzDefPoint(4.65, 2.3){C}
\tkzCircumCenter(A,B,C)\tkzGetPoint{O}
\tkzDrawArc(O,A)(C)

\tkzDefPoint(3.08, 1.2){A}\tkzDefPoint(4, 1.05){B}\tkzDefPoint(4.92, 1.2){C}
\tkzCircumCenter(A,B,C)\tkzGetPoint{O}
\tkzDrawArc(O,A)(C)   

\tkzDefPoint(3, 0.4){A}\tkzDefPoint(4, 0.25){B}\tkzDefPoint(5, 0.4){C}
\tkzCircumCenter(A,B,C)\tkzGetPoint{O}
\tkzDrawArc(O,A)(C)  

\tkzDefPoint(3, -0.4){A}\tkzDefPoint(4, -0.55){B}\tkzDefPoint(5, -0.4){C}
\tkzCircumCenter(A,B,C)\tkzGetPoint{O}
\tkzDrawArc(O,A)(C) 

\tkzDefPoint(3.08, -1.2){A}\tkzDefPoint(4, -1.35){B}\tkzDefPoint(4.92, -1.2){C}
\tkzCircumCenter(A,B,C)\tkzGetPoint{O}
\tkzDrawArc(O,A)(C)  

\tkzDefPoint(3.35, -2.25){A}\tkzDefPoint(4, -2.4){B}\tkzDefPoint(4.65, -2.25){C}
\tkzCircumCenter(A,B,C)\tkzGetPoint{O}
\tkzDrawArc(O,A)(C)

    \draw[dotted, color=black, line width=1] (10,4.4) -- (10,-4.4);

    \draw[line width=1, color=green] (10,4) -- (10,2);
    \draw[line width=1, color=red] (10,2) -- (10,-2);
    \draw[line width=1, color=blue] (10,-2) -- (10,-4);
    
\draw[-](9.8,4) -- (10.2,4);
\draw[-](9.8,2) -- (10.2,2);
    \draw[-](9.9,1.2) -- (10.1,1.2);
    \draw[-](9.9,0.4) -- (10.1,0.4);
    \draw[-](9.9,-1.2) -- (10.1,-1.2);
\draw[-](9.8,-2) -- (10.2,-2);
\draw[-](9.8,-4) -- (10.2,-4);

\tkzDefPoint(4.8,0){x}
\tkzLabelPoint[below](x){$x$}
\foreach \n in {x}
  \node at (\n)[circle,fill,inner sep=1.5pt]{};

    \draw[shift={(8.2, 1)}]   node[above, rotate=18]{using $h_1 \in F_j$};
    \draw[->](4.8,0) -- (9.8,1.6);
    \draw[shift={(10,1.6)}]   node[right]{$I_{j\cdot d_2}$};
    
    \draw[->](4.8,0) -- (9.8,0.75);
    \draw[shift={(10,0.75)}]   node[right]{$I_{j\cdot d_2+1}$};
    
    \draw[shift={(8.2, -1.18)}]   node[above, rotate=-18]{using $h_{d_2} \in F_j$};
    \draw[->](4.8,0) -- (9.8,-1.6);
    \draw[shift={(10,-1.6)}]   node[right]{$I_{(j+1)\cdot d_2-1}$};

\draw [decorate,decoration={brace,amplitude=10pt,mirror,raise=4pt},yshift=0pt]
(11.05,2) -- (11.05,4) node [black,midway,xshift=1cm] {\footnotesize
$I'_{j-1}$};

\draw [decorate,decoration={brace,amplitude=10pt,mirror,raise=4pt},yshift=0pt]
(11.05,-2) -- (11.05,2) node [black,midway,xshift=0.8cm] {\footnotesize
$I'_{j}$};

\draw [decorate,decoration={brace,amplitude=10pt,mirror,raise=4pt},yshift=0pt]
(11.05,-4) -- (11.05,-2) node [black,midway,xshift=1cm] {\footnotesize
$I'_{j+1}$};

\draw [decorate,decoration={brace,amplitude=8pt,mirror,raise=3pt},yshift=0pt]
(6,-0.5) -- (6,0.5) node [black,midway,xshift=1cm] {\footnotesize
${\mathcal G}(x)$};

\end{tikzpicture}
}
%\end{adjustbox}
\end{center}
%\vspace{-0.4cm}
\caption{Illustration for Algorithms~\ref{good-hash} and ~\ref{alg-mapping}.}
\label{fig:process}

\end{figure*}
\fi

\begin{algorithm}[htbp]
     \textbf{Input:} $|U|, k \in \nat$ where $|U| > 1$  \\
     \textbf{Output:} A family of sets of hash functions
      \begin{algorithmic}[1]
           \State let $u$ and $d$ be the unique positive integers satisfying $2^{u-1}<|U| \le 2^u$ and 
           $2^{d-1}< \frac{k}{\ln k} \le 2^{d}$\label{good-hash-1}
           \State choose $f \in_{u.a.r.} \mathcal{H}$, where $\mathcal{H} = \{ h: \{0,1\}^u \rightarrow \{0,1 \}^{d}  \}$ is a  $\lceil 12\ln k \rceil$-wise
           independent set of hash functions \label{good-hash-2}
           \State let
            $\mathcal{H'}$ be a set of universal hash functions from $U$ to $[\lceil 13\ln k \rceil^2]^-$ 
           \State let $F_i$, for $i \in [2^d]^-$, be a set of $\lceil 8\ln k \rceil$ hash functions chosen independently and u.a.r.~from $\mathcal{H'}$  \label{good-hash-6}
           \State return $\{f, F_0, \ldots, F_{2^d-1}\}$
       \end{algorithmic}
      \caption{: An algorithm for partitioning $U$ and constructing families of hash functions} \label{good-hash}
\end{algorithm}

 In the second phase, we define a relation ${\cal G}$ (from $U$) that, for each $x \in U$, associates a set ${\cal G}(x)$ of integers. This relation extends the hash functions in the $F_j$'s above by (1) ensuring that elements in different parts of $U$ (w.r.t.~the partitioning) are distinguished, in the sense that they are associated with subsets of integers that are contained in disjoint intervals of integers; and (2) maintaining the property that elements of the same part $U_j$ that are distinguished under some function in $F_j$ remain so under the extended relation. To do so, for each part $U_j$, we associate an ``offset'' and create a large gap between any two (consecutive) offsets; we will ensure that all the elements in the same $U_j$ fall within the same interval determined by two consecutive offsets. To compute the set ${\cal G}(x)$, for an element $x \in U_j$, we start with an offset $o_j$ that depends solely on $U_j$ ($o_j=j\cdot d_2\cdot d_3$ in \textbf{Algorithm \ref{alg-mapping}}),  and consider every function in the family $F_j$ corresponding to $U_j$. For each such function $h_i$, we associate an offset $o'_i$ ($o'_i=(i-1)\cdot d_3$ in \textbf{Algorithm \ref{alg-mapping}}), and for $x$ and that particular function $h_i$, we add to ${\cal G}(x)$ the value $g(j, i, x) = o_j + o'_i + h_i(x)$. The above phase is described in \textbf{Algorithm \ref{alg-mapping}}.
 
\begin{algorithm}[htbp]
     \textbf{Input:} $x \in U$, $k \in \nat$, $\{f, F_0, \ldots, F_{d_1-1}\}$  is computed by Algorithm~\ref{good-hash}, where $|F_0|=\cdots=|F_{d_1-1}|$   \\
     \textbf{Output:} a set ${\cal G}(x)$  
      \begin{algorithmic}[1]
          \State let $d_2 = |F_0|=\cdots=|F_{d_1-1}|$ and $d_3=\lceil 13\ln k \rceil^2$
           \State ${\cal G}(x)=\emptyset$
           \State compute $f(\llcorner x \lrcorner)$ and let $j$ be the integer such that 
                      $\llcorner j \lrcorner = f(\llcorner x \lrcorner)$ \label{alg-mapping-2}
           \For {$i=1$ to $d_2$}
                \State let $h_i$ be
           the $i$-th function in $F_j$ (assuming an arbitrary ordering on $F_j$) \label{alg-mapping-33}
           \State let $g(j, i, x)=j\cdot d_2 \cdot d_3+(i-1)\cdot d_3+h_i(x)$ and 
                     let ${\cal G}(x)={\cal G}(x) \cup  \{g(j, i, x)\}$
                      \label{alg-mapping-3}
        \EndFor
           \State return ${\cal G}(x)$ 
       \end{algorithmic}
      \caption{: An algorithm that defines the relation ${\cal G}$ from $U$ to $[d_1 \cdot d_2 \cdot d_3]^-$} 
      \label{alg-mapping}
\end{algorithm}

Now that the relations ${\cal G}(x)$, for $x \in U$, have been defined, we will show in the following theorem that, for any $k$-subset $S$ of $U$, \emph{w.h.p.}~there exist $k$ distinct elements $i_0, \ldots, i_{k-1}$, such that their pre-images ${\cal G}^{-1}(i_0), \ldots, {\cal G}^{-1}(i_{k-1})$ are pairwise disjoint, contain all elements of $S$, and each pre-image contains exactly one element of $S$; those pre-images serve as the desired sets $T_i$, for $i \in [k]$.

Consider \textbf{Algorithm 1} and \textbf{Algorithm 2}, and refer to them for the terminologies used in the subsequent discussions. 
Let $d_1=2^d$, $d_2=\lceil 8\ln k\rceil$ and 
$d_3=\lceil 13\ln k \rceil^2$ as defined in \textbf{Algorithms~\ref{good-hash}} and~\textbf{\ref{alg-mapping}}. For $i\in [d_1\cdot d_2 \cdot d_3]^-$, define $T_i=\{ x \in U \mid i \in {\cal G}(x)\}$.
We define next two sequences of intervals, and prove certain properties about them, that will be used in the proof of Theorem~\ref{alg-hash-THM4}.
For $q\in [d_1\cdot d_2]^-$, let $I_q = \{ r \mid q\cdot d_3 \le r < (q+1)\cdot d_3\}$. For $t \in [d_1]^-$, let $I'_t=\{ r \mid t\cdot d_2 \cdot d_3 \le r < t\cdot d_2 \cdot d_3+ d_2\cdot d_3 \}$. 
Note that each interval $I'_t$ is partitioned into the $d_2$-many intervals $I_q$, for $q =t\cdot d_2, \ldots, t\cdot d_2 + d_2-1$.

\begin{lemma}\ifshort {\rm ($\spadesuit$)} \fi  
\label{lem:claims}
The following statements hold:

\begin{itemize}
\item[(A)] For any two distinct integers $a, b \in I_q$, where $q \in [d_1\cdot d_2]^-$, we have $T_a \cap T_b= \emptyset$.

\item[(B)] For $t\in [d_1]^-$, we have ${\cal G}(U_t) \subseteq I'_t$. Moreover, for any  $a\in I'_t, b\in I'_s$, where $s\neq t$, we have $T_a \cap T_b = \emptyset$. 
\end{itemize}
\end{lemma}
\iflong
\begin{proof}

To prove (A), we
proceed by contradiction. Assume that there exists $x \in U$ such that $x \in T_a \cap T_b$. This implies that both $a$ and $b$ are in ${\cal G}(x)$. Without loss of generality, assume that \textbf{Algorithm \ref{alg-mapping}} adds $a$ to ${\cal G}(x)$ in iteration $i_a$ of Steps \ref{alg-mapping-33}--\ref{alg-mapping-3} and adds $b$ in iteration $i_b$, where $i_a < i_b$. This implies that $a = j\cdot d_2 \cdot d_3+(i_a-1)\cdot d_3+h_{i_a}(x)$ and  $b = j\cdot d_2 \cdot d_3+(i_b-1)\cdot d_3+h_{i_b}(x)$, where $j$ is the integer such that $\llcorner j \lrcorner = f(\llcorner x \lrcorner)$. Since $h_{i_a}(x)<d_3$ and 
$h_{i_b}(x) < d_3$,
it follows that $a$ belongs to the interval $I_{j\cdot d_2+i_a-1}$ and $b$ belongs to the interval $I_{j\cdot d_2+i_b-1}$, which are two distinct intervals (since  $i_a < i_b$), contradicting the assumption that both $a, b \in I_q$.

To prove (B), 
let $t\in [d_1]^-$, and let $x \in U_t$. The set ${\cal G}(x)$ consists of the values $g(t, i, x)= t\cdot d_2 \cdot d_3+(i-1)\cdot d_3+h_i(x)$, for $i=1, \ldots, d_2$. 
Since $0\le h_i(x)< d_3$ for any $i \in [d_2]$ and any $x \in U$,
it follows that $t\cdot d_2 \cdot d_3 \le g(t, i, x) \le t \cdot d_2 \cdot d_3+(d_2-1)\cdot d_3+d_3-1< t\cdot d_2 \cdot d_3+d_2 \cdot d_3$, and hence, $g(t, i, x) \in I'_t$. This proves that ${\cal G}(U_t) \subseteq I'_t$.

To prove that, for any  $a\in I'_t, b\in I'_s$, where $s\neq t$, we have $T_a \cap T_b = \emptyset$, suppose not and let $x \in T_a \cap T_b$. By the first part of the claim, we have $a \in {\cal G}(x) \subseteq I'_t$ and $b \in {\cal G}(x) \subseteq I'_s$, which is a contradiction since $I'_t \cap I'_s=\emptyset$.
\end{proof}
 \fi

\begin{theorem}  \label{alg-hash-THM4}
 For any subset $S \subseteq U$ of cardinality $k \ge 2$, with probability at least $1-\frac{4}{k^3\ln k}$, there exist $k$ sets $T_{i_0}, \ldots, T_{i_{k-1}}$ such that: (1) $|T_{i_j}\cap S|=1$ for $j\in [k]^-$, (2) $S\subseteq \cup_{j\in [k]^-} T_{i_j}$, and 
(3) $T_{i_j} \cap T_{i_l} = \emptyset$ for $j \neq l \in [k]^-$. 
\end{theorem}

\begin{proof}
For $j\in [d_1]^-$, let $U_j$ be the set of elements in $U$ whose image is $\llcorner j \lrcorner$ under $f$ (defined in Step~\ref{good-hash-2} of \textbf{Algorithm~\ref{good-hash}}), that is  $U_j= \{y\in U \mid f(\llcorner y \lrcorner)=\llcorner j \lrcorner\}$. Clearly, the sets $U_j$, for $j\in [d_1]^-$, partition the universe $U$. We will show that, with probability at least $1-\frac{4}{k^3\ln k}$, there exist $k$ sets $T_{i_0}, \ldots, T_{i_{k-1}}$ 
that satisfy conditions (1)--(3) in the statement of the theorem.

Let $S \subseteq U$ be any subset such that $|S|=k$. For $ j \in [d_1]^-$ and $y \in S$, let $X_{y, j}$ be the random variable defined as $X_{y, j}= 1$ if $f(\llcorner y \lrcorner)=\llcorner j \lrcorner$ and 0 otherwise. Let
$X_j = \sum_{y \in S} X_{y, j}$, and $S_j=\{y\in S \mid f(\llcorner y \lrcorner)=\llcorner j \lrcorner\}$. Thus, $|S_j|=X_j$. 
Since $f$ is $\lceil 12 \ln k \rceil$-wise independent, 
the random variables $X_{y, j}$, for $y\in S$, are $\lceil 12 \ln k \rceil$-wise independent and 
$\Pr(X_{y, j}=1)=\frac{1}{d_1}$. Thus, $E[X_j]=|S| \cdot \frac{1}{d_1}$. 
Since $d_1=2^d$ and $2^{d-1}<\frac{k}{\ln k} \le 2^d$ by definition,
 we have $\frac{k}{\ln k} \le d_1 < \frac{2k}{\ln k}$ and 
$ \frac{\ln k}{2} < E[X_j] \le \ln k$. Applying Theorem \ref{chernoff-bound} with 
$\mu = E[X_j]$ and $\delta=\frac{12\ln k}{E[X_j]}>1$, we get 
$\Pr(X_j \ge (1+\delta) E[X_j])  \le e^{-E[X_j]\delta/3}=\frac{1}{k^4}$. 
Since $E[X_j]\le \ln k$ and $\delta=\frac{12\ln k}{E[X_j]}$, we have
$(1+\delta)E[X_j]\le 13 \ln k$. Hence, $\Pr(X_j \ge 13\ln k) \le \Pr(X_j \ge (1+\delta)E[X_j]) \le \frac{1}{k^4}$. 
Let $\mathcal{E}$ denote the event that 
$\bigwedge_{i \in [d_1]^-} (X_i \le 13\ln k)$. 
By the union bound, we have $\Pr(\mathcal{E}) \ge 1-\frac{d_1}{k^4} 
\ge 1-\frac{2}{k^3\ln k}$, where the last inequality holds since $d_1 < 2k/\ln{k}$. 

Assume that event $\mathcal{E}$ occurs, i.e., that $|S_j|\le 13\ln k$ holds for $j\in [d_1]^-$. 
Consider Step \ref{good-hash-6} in \textbf{Algorithm \ref{good-hash}}. 
Fix $j\in [d_1]^-$, and let $E_j$ be the event that $F_j$ does not contain any perfect hash function w.r.t.~$S_j$. 
Let $h$ be a hash function picked from $\mathcal{H'}$ u.a.r. 
Since $|S_j| \le 13\ln k$ (by assumption), by 
Theorem \ref{lemma-hash}, with probability at least $1/2$, $h$ is perfect w.r.t.~$S_j$. Since $F_j$ consists of $\lceil 8\ln k \rceil$ hash functions chosen 
independently and u.a.r.~from $\mathcal{H'}$, we have $\Pr(E_j) \le (1/2)^{\lceil 8\ln k \rceil} < \frac{1}{k^4}$. 
Applying the union bound, we have $\Pr(\cup_{j\in [d_1]^-} E_j) \le \frac{d_1}{k^4}< \frac{2}{k^3\ln k}$. 
Let $\mathcal{E}'$ be the event that there exist 
$d_1$ functions $f_0, f_1, \ldots, f_{d_1-1}$ such that $f_j \in F_j$ and $f_j$ is perfect 
w.r.t.~$S_j$, $j\in [d_1]^-$. Therefore, 
$\Pr(\mathcal{E}') \ge \Pr(\mathcal{E})(1-\Pr(\cup_{j\in [d_1]^-} E_j)) \ge 1-\frac{4}{k^3\ln k} + \frac{4}{k^6\ln^2 k}  \ge 1-\frac{4}{k^3\ln k}$. 
Suppose that such a set $\{f_0, \ldots, f_{d_1-1}\}$ of functions exists. 
Let $\eta(q)$ be the iteration number $i$ in Step \ref{alg-mapping-33} of \textbf{Algorithm \ref{alg-mapping}} during 
which $f_q \in F_q$ is chosen, for $q\in [d_1]^-$. We define the following (multi-)set $B$ as follows. For each $q \in [d_1]^-$, and for element $x \in S_q$, add to $B$ the element $g(q,\eta(q),x))$ defined in Steps \ref{alg-mapping-33}--\ref{alg-mapping-3} of \textbf{Algorithm \ref{alg-mapping}} (by $\{f, f_0, \ldots, f_{k-1}\}$). Observe that, by the definition of $B$, for every $x \in S$, there exists $a \in B$ such that $x \in T_a$.
We will show next that $B$ contains exactly $k$ distinct elements, and that, for any $a\neq b \in B$,
%for any two distinct elements $a, b \in B$, 
it holds that $T_a \cap T_b = \emptyset$. The above will show that the sets $\{T_a \mid a \in B\}$ satisfy conditions (1)--(3) of the theorem, thus proving the theorem. 

It suffices to show that for any two distinct elements of $S$, the corresponding elements added to $B$ are distinct. 
Let $x_1$ and $x_2$ be two distinct elements of $S$. Assume that $x_1\in S_j$ and $x_2 \in S_l$, where $j, l \in [d_1]^-$. We distinguish two cases based on whether or not $j=l$.

If $j=l$, we have $g(j, \eta(j), x_1)=j\cdot d_2 \cdot d_3+(\eta(j)-1)\cdot d_3 + f_j(x_1)$ and 
$g(j, \eta(j), x_2)= j \cdot d_2 \cdot d_3 +(\eta(j)-1)\cdot d_3+f_j(x_2)$. 
Since $f_j$ is perfect w.r.t.~$S_j$, we have
$g(j, \eta(j), x_1) \neq g(j, \eta(j), x_2)$. 
Moreover, both $g(j, \eta(j), x_1)$ and $g(j, \eta(j), x_2)$ are in  $I_{j\cdot d_2+(\eta(j)-1)}$ (since $0 \le h_j(x_1), h_j(x_2)<d_3$), where  $j\cdot d_2+(\eta(j)-1) \le (d_1-1) \cdot d_2 + (d_2-1) \in [d_1 \cdot d_2]^-$. By part (A) of Lemma \ref{lem:claims}, it holds that $T_{g(j, \eta(j), x_1)} \cap T_{g(j, \eta(j), x_2)} = \emptyset$.

Suppose now that $j \neq l$. By definition of $S_j, S_l, U_j, U_l$, we have 
$S_j \subseteq U_j$ and $S_l \subseteq U_l$. Consequently, 
 $g(j, \eta(j), x_1) \in {\cal G}(U_j)$ and
 $g(l, \eta(l), x_2) \in {\cal G}(U_l)$ hold. By part (B) of Lemma \ref{lem:claims}, we have ${\cal G}(U_j) \subseteq I'_j$ 
 and ${\cal G}(U_l) \subseteq I'_l$. Therefore, $g(j, \eta(j), x_1) \neq g(l, \eta(l), x_2)$.
 Moreover, $T_{g(j, \eta(j), x_1)} \cap T_{g(l, \eta(l), x_2)} = \emptyset$ holds by part (B) of Lemma \ref{lem:claims} as well. 
\end{proof}

\begin{theorem}\ifshort {\rm ($\spadesuit$)} \fi  \label{alg-hash-space}
\textbf{Algorithm \ref{good-hash}} runs in space $\bigoh(k+(\log k)(\log |U|))$, and 
\textbf{Algorithm \ref{alg-mapping}} runs in space $\bigoh(\log k)$ and in time polynomial in 
$\log |U|$. 
\end{theorem}
\iflong
\begin{proof}
In \textbf{Algorithm \ref{good-hash}}, since $f$ is $\lceil 12\ln k \rceil$-wise independent, 
storing $f$ uses space $\bigoh(\ln k \cdot \max\{u,d \})=\bigoh((\log k)(\log |U|))$ (since $k\le |U|$) by Theorem \ref{evaluate-time}. 
Storing a universal hash function 
uses $\bigoh(1)$ space, and thus storing $\{F_0, \ldots, F_{d_1-1}\}$ 
uses $\bigoh(d_1 \cdot d_2)=\bigoh(k)$ space. Therefore, \textbf{Algorithm \ref{good-hash}}
can be implemented in space $\bigoh(k+(\log k)(\log |U|))$. 

For \textbf{Algorithm~\ref{alg-mapping}}, since ${\cal G}(x)$ contains exactly $d_2$ elements, storing ${\cal G}(x)$ takes
$\bigoh(d_2)=\bigoh(\ln k)$ space.
 In Step \ref{alg-mapping-2}, computing $f(\llcorner x \lrcorner)$
 takes time polynomial in $\log |U|$ and $\log k$ by Theorem \ref{evaluate-time}, 
 since $f$ is a $\lceil 12\ln k \rceil$-wise independent hash function from $\{0,1\}^u$ to $\{0,1\}^d$.
 Computing $j$ in Step \ref{alg-mapping-2} takes time polynomial in $d=\bigoh(\log k)$ since 
$f(\llcorner x \lrcorner) \in \{0,1\}^d$. Therefore, Step \ref{alg-mapping-2} can be performed in time 
polynomial in $\log |U|$ and $\log k$, and hence polynomial in $\log |U|$ (since $k \le |U|$). 
Step \ref{alg-mapping-3} can be implemented in time polynomial in $\log k$, since $|F_j|= \lceil 8\ln k \rceil$. 
Altogether, \textbf{Algorithm \ref{alg-mapping}} takes time polynomial in $\log |U|$. This completes the proof. 
\end{proof}
\fi

\section{Dynamic Streaming Model}
\label{sec:dynamic}
In this section, we present results on p-{\sc Matching} and p-{\sc WT-Matching} in the dynamic streaming model. 
%First, we present a sampling algorithm \textbf{Alg-Sampling} that \emph{w.h.p.}~returns the edges of a maximum-weight $k$-matching in $G$ (if a $k$-matching in $G$ exists). 
The algorithm uses the toolkit developed in the previous section, together with the $\ell_0$-sampling technique discussed in Section~\ref{sec:prelim}. We first give a high-level description of how the algorithm works.  

We will hash the vertices of the graph to a range $R$ of size $\bigoh(k \log^2k)$. For each element $(e=uv, wt(e), op) \in \SSS$, we use the relation ${\cal G}$, discussed in Section~\ref{sec:structural}, and compute the two sets ${\cal G}(u)$ and ${\cal G}(v)$. For each $i \in {\cal G}(u)$ and each $j \in {\cal G}(v)$, we associate an instance of an $\ell_0$-sampler primitive, call it ${\mathscr{C}}_{i,j,wt(uv)}$, and update it according to the operation $op$. Recall that it is assumed that the weight of every edge does not change throughout the stream.  

The solution computed by the algorithm consists of a set of edges created by invoking each of the $\tilde{O}(Wk^2)$ $\ell_0$-sampler algorithms to sample at most one edge from each ${\mathscr{C}}_{i,j,w}$, for each pair of $i, j$ in the range $R$ and each edge-weight of the graph stream. 

The intuition behind the above algorithm (i.e., why it achieves the desired goal) is the following. Suppose that there exists a maximum-weight $k$-matching $M$ in $G$, and let $M=\{u_0u_1, \ldots, u_{2k-2}u_{2k-1}\}$. By Theorem~\ref{alg-hash-THM4}, \emph{w.h.p.}~there exist $i_0, \ldots, i_{2k-1}$ in the range $R$ such that $u_{j} \in T_{i_j}$, for $j \in [2k-1]^-$, and such that the $T_{i_j}$'s are pairwise disjoint. Consider the $k$ $\ell_0$-samplers $\mathscr{C}_{i_{2j},i_{2j+1},wt(u_{2j}u_{2j+1})}$, where $j\in [k]^-$. Then, \emph{w.h.p.}, the $k$ edges sampled from these $k$ $\ell_0$-samplers are the edges of a maximum-weight $k$-matching (since the $T_{i_j}$'s are pairwise disjoint) whose weight equals that of $M$.

\begin{algorithm}[htbp]
\caption{The streaming algorithm $\mathcal{A}_{dynamic}$ in the dynamic streaming model}\label{alg-postprocessing}
%\textbf{Input:} $n=|V(G)|$ and a parameter $k \in \nat$  

\begin{algorithmic}[1]

\vspace*{3mm}
\Algphase{{\bf $\mathcal{A}_{dynamic}$-Preprocess}: The preprocessing algorithm}
\Require $n=|V(G)|$ and a parameter $k \in \nat$

\State let $\mathscr{C}$ be a set of $\ell_0$-sampling primitive instances and  $\mathscr{C}=\emptyset$ \label{preprocess3}
       \State let $\{f, F_0, F_1, \ldots,$ $F_{d_1-1}\}$ be the output of \textbf{Algorithm \ref{good-hash}} on input $(n, 2k)$ \label{preprocess1}
 %      \State return $\{f, F_0, F_1, \ldots,$ $F_{d_1-1}, \mathscr{C}\}$ \label{preprocess2}
 
\vspace*{3mm}
\Algphase{{\bf $\mathcal{A}_{dynamic}$-Update}: The update algorithm}
\Require The $i$-th update $(e_i=uv, wt(e), op) \in \SSS$
\setcounter{ALG@line}{0}

\State let ${\cal G}(u)$ be the output of \textbf{Algorithm \ref{alg-mapping}} on input ($u, 2k, \{f, F_0, F_1, \ldots, F_{d_1-1} \}$) \label{alg-sampling-5}
     \State let ${\cal G}(v)$ be the output of \textbf{Algorithm \ref{alg-mapping}} on input ($v, 2k, \{f, F_0, F_1, \ldots, F_{d_1-1} \}$) \label{alg-sampling-6}
    \For { $i\in {\cal G}(u)$ and $j \in {\cal G}(v)$}   \label{alg-sampling-7}
          \If { $\mathscr{C}_{i,j,wt(uv)} \notin \mathscr{C}$ } \label{alg-sampling-55}
                \State create the $\ell_0$-sampler $\mathscr{C}_{i,j,wt(uv)}$ \label{alg-sampling-66}
            \EndIf
           \State feed $\langle uv,op \rangle$ to the $\ell_0$-sampling algorithm $\mathscr{C}_{i,j,wt(uv)}$ with parameter $\delta$ \label{alg-sampling-8}
      \EndFor
      
\vspace*{3mm}
\Algphase{{\bf $\mathcal{A}_{dynamic}$-Query}: The query algorithm after the $i$-th update }

\setcounter{ALG@line}{0}

     \State let $E'=\emptyset$ \label{alg-sampling-99}
\For { each $\mathscr{C}_{i,j,w} \in \mathscr{C}$}  \label{alg-sampling-9}
                    \State apply the $\ell_0$-sampler $\mathscr{C}_{i, j, w}$ with parameter $\delta$ to sample an edge $e$ \label{alg-sampling-10}
                    \State if $\mathscr{C}_{i, j, w}$ does not FAIL then set $E'=E' \cup \{e\}$ \label{alg-sampling-11}
       \EndFor
       \State return a maximum-weight $k$-matching in $G'=(V(E'), E')$ if any; otherwise, return $\emptyset$ \label{alg-sampling-12}  
       
\end{algorithmic}
\end{algorithm}

Let $\SSS$ be a graph stream of a weighted graph $G=(V, E)$ with $W$ distinct weights, where $W \in \nat$, and let $n=|V|$ and $k \in \nat$. Choose $\delta  = \frac{1}{20k^4\ln (2k)}$. Let ${\cal A}_{dynamic}$ be the algorithm consisting of the sequence of three subroutines/algorithms  \textbf{$\mathcal{A}_{dynamic}$-Preprocess},  \textbf{$\mathcal{A}_{dynamic}$-Update}, and \textbf{$\mathcal{A}_{dynamic}$-Query}, where \textbf{$\mathcal{A}_{dynamic}$-Preprocess} is applied at the beginning of the stream, \textbf{$\mathcal{A}_{dynamic}$-Update} is applied after each operation, and \textbf{$\mathcal{A}_{dynamic}$-Query} is applied whenever the algorithm is queried for a solution after some update operation. Without loss of generality, and for convenience, we will assume that the algorithm is queried at the end of the stream $\SSS$, even though the query could take place after any arbitrary operation.

\begin{lemma}  \ifshort {\rm ($\spadesuit$)} \fi \label{A_reduce_update}
Let $M'$ be the matching obtained by applying the algorithm ${\cal A }_{dynamic}$ with \textbf{$\mathcal{A}_{dynamic}$-Query} invoked at the end of $\SSS$. If $G$ contains a $k$-matching then, with probability at least $1-\frac{11}{20k^3\ln (2k)}$, $M'$ is a  maximum-weight $k$-matching of $G$.
\end{lemma}
\iflong
\begin{proof}
Suppose that $G$ has a $k$-matching, 
and let  $M=\{u_0u_1, \ldots, u_{2k-2}u_{2k-1}\}$ be a maximum-weight $k$-matching in $G$. 
(Note that we can assume that $u_{2j}<u_{2j+1}$ for every $j\in [k]^-$; see Section~\ref{sec:prelim}.) 
From  \textbf{Algorithm \ref{good-hash}} and 
\textbf{Algorithm \ref{alg-mapping}}, 
it follows that $d_1=O(\frac{k}{\ln k})$, 
$d_2=O(\ln k)$, and $d_3=O(\ln^2 k)$.
For $i\in [d_1\cdot d_2 \cdot d_3]^-$, 
let $T_i=\{ u \in V \mid i \in {\cal G}(u)\}$.
By Theorem \ref{alg-hash-THM4}, with probability at least $1-\frac{4}{(2k)^3 \ln (2k)} = 1- \frac{1}{2k^3\ln (2k)}$,
there exist $i_0, i_1, \ldots, i_{2k-1}$ such that (1) $u_j \in T_{i_j}, j\in [2k]^-$, and (2) 
$T_{i_j} \cap T_{i_l} = \emptyset$ for $j\neq l \in [2k]^-$. Let $\mathcal{E'}$ be the above 
event. Then $\Pr(\mathcal{E'}) \ge 1- \frac{1}{2k^3\ln (2k)}$.
By Step \ref{alg-sampling-8} of \textbf{$\mathcal{A}_{dynamic}$-Update}, $u_{2j}u_{2j+1}$ will be fed into  
$\mathscr{C}_{i_{2j}, i_{2j+1}, wt(u_{2j}u_{2j+1})}$ for
$j\in [k]^-$. Hence, $\mathscr{C}_{i_{2j}, i_{2j+1}, wt(u_{2j}u_{2j+1})}$ is fed at least one edge for every $j\in [k]^-$.

Now, let us compute the sampling success probability 
in Step \ref{alg-sampling-11} of \textbf{$\mathcal{A}_{dynamic}$-Query}. 
Note that this probability involves %the success probabilities of 
both Step \ref{alg-sampling-8} of 
\textbf{$\mathcal{A}_{dynamic}$-Update} 
and Step  \ref{alg-sampling-11} of \textbf{$\mathcal{A}_{dynamic}$-Query}. In Step \ref{alg-sampling-8} of \textbf{$\mathcal{A}_{dynamic}$-Update} and Step  \ref{alg-sampling-11} of \textbf{$\mathcal{A}_{dynamic}$-Query}, we employ the $\ell_0$-sampling primitive in Lemma \ref{lem:sampler}. Let $\mathcal{E}$ be the event that one edge is sampled successfully for each $\ell_0$-sampler in 
$\{ \mathscr{C}_{i_{2j}, i_{2j+1}, wt(u_{2j}u_{2j+1})} \mid j\in [k]^- \}$. By Lemma \ref{lem:sampler}, one 
$\ell_0$-sampler fails with probability at most $\delta$. Hence, 
 $\Pr(\mathcal{E}) \ge 1 - k\cdot \delta$ by the union bound. Since 
 $\delta = \frac{1}{20k^4\ln (2k)}$, we get $\Pr(\mathcal{E}) \ge 1 - \frac{1}{20k^3\ln (2k)}$. 
Hence,
with probability at least $1-\frac{1}{20k^3\ln (2k)}$, 
$E'$ will contain one edge $e_j$ sampled from $\mathscr{C}_{i_{2j},i_{2j+1}, wt(u_{2j}u_{2j+1})}$ for each 
$j\in [k]^-$.
Note that $e_j$ may be $u_{2j}u_{2j+1}$ or any other edge with the same weight as $u_{2j}u_{2j+1}$. 
Since $T_{i_j} \cap T_{i_l}=\emptyset$ for every $j\neq l \in [2k]^-$, we have that the edges fed to 
$\mathscr{C}_{i_{2a},i_{2a+1}, wt(u_{2a}u_{2a+1})}$ and 
$\mathscr{C}_{i_{2b},i_{2b+1}, wt(u_{2b}u_{2b+1})}$ are 
vertex disjoint, 
for all $a \neq b \in [k]^-$. 
Thus, $\{ e_0, \ldots, e_{k-1}\}$ forms a maximum-weight $k$-matching of $G$.  Applying the union bound, the probability that the graph $G'$ contains a maximum-weight $k$-matching of $G$ is at least 
$1-\Pr(\bar{\mathcal{E}} \cup \bar{\mathcal{E'}})\ge $
$1 -\Pr(\bar{\mathcal{E'}}) - \Pr(\bar{\mathcal{E}}) \geq 1-\frac{1}{2k^3\ln (2k)}-\frac{1}{20k^3\ln (2k)} = 1-\frac{11}{20k^3\ln (2k)}$.
Since $M'$ is a maximum-weight $k$-matching of $G'$, $M'$ is a maximum-weight $k$-matching of $G$ as well. 
\end{proof}
\fi
 
\begin{theorem} \label{theorem7}
The algorithm ${\cal A}_{dynamic}$ outputs a matching $M'$ such that (1) if $G$ contains a $k$-matching then, with probability at least  $1-\frac{11}{20k^3\ln (2k)}$, $M'$ is a maximum-weight $k$-matching of $G$; and (2) if $G$ does not contain a $k$-matching
then $M'=\emptyset$.   Moreover, the algorithm ${\cal A}_{dynamic}$ runs in $\tilde{O}(Wk^2)$ space and has $\tilde{O}(1)$ update time.
\end{theorem}
\begin{proof}

First, observe that $G'$ is a subgraph of $G$, since it consists of edges sampled from subsets of edges in $G$. Therefore, statement (2) in the theorem clearly holds true. Statement (1) follows from Lemma \ref{A_reduce_update}. Next, we analyze the update time of algorithm ${\cal A}_{dynamic}$. 

From \textbf{Algorithm \ref{good-hash}} and 
\textbf{Algorithm \ref{alg-mapping}}, 
we have $d_1=O(\frac{k}{\ln k})$, 
$d_2=O(\ln k)$, $d_3=O(\ln^2 k)$ and $|F_i|=O(\ln k)$ for $i\in [d_1]^-$. Thus, 
$|{\cal G}(u)|=O(\ln k)$ holds for all $u\in V$. 
For the update time, it suffices to examine Steps \ref{alg-sampling-5}--\ref{alg-sampling-8} of \textbf{$\mathcal{A}_{dynamic}$-Update}. 
By Theorem \ref{alg-hash-space}, Steps~\ref{alg-sampling-5}--\ref{alg-sampling-6} take time
polynomial in $\log n$, which is $\tilde{O}(1)$.
For Step \ref{alg-sampling-55}, we can index $\mathscr{C}$ using
a sorted sequence of triplets $(i,j,w)$, 
where $i, j\in [d_1\cdot d_2 \cdot d_3]^-$ and $w$ ranges over all possible weights. 
Since $d_1=\bigoh(\frac{k}{\ln k})$, $d_2=\bigoh(\ln k)$ and $d_3=\bigoh(\ln^2 k)$, 
we have $|\mathscr{C}|=\bigoh((d_1\cdot d_2 \cdot d_3)^2 \cdot W)=\bigoh(Wk^2\ln^4 k)$. 
Using binary search on $\mathscr{C}$, one execution of Step \ref{alg-sampling-55} takes time $\bigoh(\log W+ \log k)$. 
Since $|{\cal G}(u)|=O(\ln k)$ for every $u\in V$, and since by Lemma~\ref{lem:sampler} updating the sketch for an $\ell_0$-sampler takes $\tilde{\bigoh{}}(1)$ time, Steps \ref{alg-sampling-7}--\ref{alg-sampling-8} take time $\bigoh(\ln^2 k) \cdot (\bigoh(\log W+ \log k)+\tilde{O}(1))= \tilde{O}(1)$. 
Therefore, the overall update time is $\tilde{O}(1)$. 

Now, we analyze the space complexity of the algorithm. 
First, consider \textbf{$\mathcal{A}_{dynamic}$-Preprocess}. 
Obviously, Step \ref{preprocess3} uses $\bigoh(1)$ space. 
Steps \ref{preprocess3}--\ref{preprocess1} use space $\bigoh(k+(\log k)(\log n))$ 
(including the space used to store 
$\{f, F_0, \ldots, F_{d_1-1}, \mathscr{C}\}$) by Theorem \ref{alg-hash-space}. 
Altogether, \textbf{$\mathcal{A}_{dynamic}$-Preprocess} runs in space $\bigoh(k+(\log k)(\log n))$. 
Next, we discuss \textbf{$\mathcal{A}_{dynamic}$-Update}. 
Steps \ref{alg-sampling-5}--\ref{alg-sampling-6}
 take space $O(\ln k)$ by Theorem \ref{alg-hash-space}. 
Observe that the space used in Steps \ref{alg-sampling-7}--\ref{alg-sampling-8}
is dominated by the space used by the set $\mathscr{C}$ of 
$\ell_0$-sampling primitive instances. 
By Lemma \ref{lem:sampler}, one instance of an $\ell_0$-sampling 
primitive uses space $\bigoh(\log^2 n \cdot \log (\delta^{-1}))$. Since $\delta = \frac{1}{20k^4\ln (2k)}$, we have $\log (\delta^{-1})=\bigoh(\log k)$. It follows that a single
instance of an $\ell_0$-sampler uses space 
$\bigoh(\log^2 n \cdot \log k)$. 
Since $|\mathscr{C}|=\bigoh(Wk^2\ln^4 k)$, Steps~\ref{alg-sampling-7}--\ref{alg-sampling-8} use space 
$\bigoh(Wk^2\log^2n \log^5 k)=\tilde{\bigoh}(Wk^2)$. Finally, consider \textbf{${\cal A}_{dynamic}$-Query}.
The space in Steps \ref{alg-sampling-99} -- \ref{alg-sampling-11} is dominated by the space used by $\mathscr{C}$ and the space needed to store the graph $G'$, and hence $E'$. 
By the above discussion, $\mathscr{C}$ takes space $\tilde{\bigoh}(Wk^2)$.
Since at most one edge is sampled from each $\ell_0$-sampler instance and $|\mathscr{C}|=\bigoh(Wk^2\ln^4 k)$, 
we have $|E'|=|\mathscr{C}|=\bigoh(Wk^2\ln^4 k)$. Step \ref{alg-sampling-12} utilizes space $\bigoh(|E'|)$ \cite{gabow1, gabow2}.
Therefore, 
\textbf{${\cal A}_{dynamic}$-Query} runs in space $\tilde{\bigoh}(Wk^2)$. It follows that the space complexity of $\mathcal{A}_{dynamic}$ is $\tilde{\bigoh}(Wk^2)$. 
\end{proof}

%\todo[inline]{The following 3 paragraphs: p-WT-Matching or Maximum matching?}

\iflong The space complexity of the above algorithm is large if the number of distinct weights $W$ is large. Under the same promise that the parameter $k$ is at least as large as the size of any maximum matching in $G$, an approximation scheme for \wmatcheq{} that is more space efficient was presented in~\cite{Chitnis2016}. This scheme approximates \wmatcheq{} to within ratio $1 + \epsilon$,  for any $\epsilon > 0$, and has space complexity $\tilde{O}(k^2\epsilon^{-1}\log W')$  and update time $\tilde{O}(1)$. The main idea behind this approximation scheme is to reduce the number of distinct weights in $G$ by rounding 
each weight to the nearest power of $1+\epsilon$. \fi

Using Theorem~\ref{theorem7}, and following the same approach in~\cite{Chitnis2016}, we can obtain the same approximation result as in~\cite{Chitnis2016}, albeit without the reliance on such a strong promise:

\begin{theorem} \ifshort {\rm ($\spadesuit$)} \fi  \label{theorem8}
Let $\SSS$ be a graph stream of a graph $G$, and let $W'=wt(e)/wt(e')$, where $e \in E(G)$ is an edge with the maximum weight and $e' \in E(G)$ is an edge with the minimum weight. Let $0<\epsilon<1$. 
In the dynamic streaming model, there exists an algorithm for p-{\sc WT-Matching} that computes a matching $M'$ 
such that (1) if $G$ contains a maximum-weight $k$-matching $M$, then 
with probability at least $1-\frac{11}{20k^3\ln (2k)}$, 
 $wt(M')>(1-\epsilon)wt(M)$; and (2) if $G$ does not contain a $k$-matching
then $M'=\emptyset$.  
Moreover, the algorithm runs in $\tilde{O}(k^2\epsilon^{-1}\log W')$ space and has $\tilde{O}(1)$ update time.
\end{theorem}
%\iflong
\begin{proof}
For each edge $e\in E$, round $wt(e)$ and assign it a new weight of $(1+\epsilon)^i$ such that $(1+\epsilon)^{i-1}<wt(e)\le (1+\epsilon)^i$. Thus, there are $\bigoh(\epsilon^{-1}\log W')$ 
distinct weights after rounding. By Theorem \ref{theorem7}, the space and update time are 
$\tilde{\bigoh}(k^2\epsilon^{-1}\log W')$ and $\tilde{\bigoh}(1)$ respectively, and the success probability is at least 
$1-\frac{11}{20k^3\ln(2k)}$. Now we  prove that $wt(M')>(1-\epsilon)wt(M)$. 

Let $e \in M$ and let
$e'$ be the edge sampled from the $\ell_0$-sampler that $e$ is fed to. 
It suffices to prove that $wt(e')>(1-\epsilon)wt(e)$. Assume that $wt(e)$ is rounded to $(1+\epsilon)^i$. 
Then, $wt(e')$ is rounded to 
$(1+\epsilon)^i$ as well. If $wt(e') \ge wt(e)$, we are done; otherwise, 
$(1+\epsilon)^{i-1}<wt(e')<wt(e)\le (1+\epsilon)^i$. It follows that 
$wt(e')>(1+\epsilon)^{i-1}\ge wt(e)/(1+\epsilon)>(1-\epsilon)wt(e)$. 
\end{proof}
%\fi
The following theorem is a consequence of Theorem~\ref{theorem7} (applied with $W=1$):

\begin{theorem} 
In the dynamic streaming model, there is an algorithm for p-{\sc Matching} such that, on input $(\SSS, k)$, the algorithm outputs a matching $M'$ satisfying that (1) if $G$ contains a  $k$-matching %$M$ 
then, with probability at least $1-\frac{11}{20k^3\ln (2k)}$, $M'$ is a $k$-matching of $G$; and (2) if $G$ does not contain a $k$-matching 
then $M'=\emptyset$.  
Moreover, the  algorithm runs in 
$\tilde{\bigoh}(k^2)$ space and has $\tilde{\bigoh}(1)$ update time. 
\end{theorem}

\iflong

\subsection{Lower Bound}
\label{subsec:lowerbounds}

Consider an undirected graphs $G=(V, E)$ with a weight function $wt:E(G) \longrightarrow \mathbb{R}_{\ge 0}$. 
We define a more general dynamic graph streaming model for an undirected graph $G$: $G$ is given as a stream 
$\SSS=(e_{i_1}, \Delta_1(e_{i_1})), \ldots, (e_{i_j}, \Delta_j(e_{i_j})), \ldots$ of updates of the weights of 
the edges, where $e_{i_j}$ 
is an edge and $\Delta_j(e_{i_j})\in \mathbb{R}$, and a parameter $k \in \nat$. 
The $j$-th update $(e_{i_j}, \Delta_j(e_{i_j}))$ updates the weight of $e_{i_j}$ by setting 
$wt(e_{i_j})=wt(e_{i_j})+\Delta_j(e_{i_j})$. We assume that 
$wt(\cdot) \ge 0$ for every update $j$. Initially, $wt(\cdot)=\textbf{0}$. 
This models allows the weight of an edge to dynamically change,  
and generalizes the dynamic graph streaming model in~\cite{Chitnis2016}, 
where $wt(e)$ is either $0$ or a fixed 
value associated with $e$. In particular, each element $(e_{i_j}, \Delta_j(e_{i_j}))$ in $\SSS$ 
is either $(e_{i_j}, wt(e_{i_j}))$ or $(e_{i_j}, -wt(e_{i_j}))$, 
and $(e_{i_j}, wt(e_{i_j}))$ means to insert the edge $e_{i_j}$
while $(e_{i_j}, -wt(e_{i_j}))$ represents the deletion of the 
edge $e_{i_j}$. 

In this subsection, we prove a lower bound for the weighted $k$-matching problem in the more general dynamic streaming model.
This lower bound result holds \emph{even} for parameter value $k=1$. We prove this lower bound via a reduction from the problem of computing the function $F_{\infty}$ of data streams defined as follows:
 
Given a data stream $\SSS' = x_1, x_2, \ldots, x_m$, where each $x_i \in \{1, \ldots,   n'\}$, let $c_i=|\{ j \mid x_j = i \}|$ denote the number of occurrences of $i$ in the stream $\SSS'$. Define 
$F_{\infty} = \max_{1\le i \le n'} c_i$. The following theorem appears in~\cite{Tim}:

\begin{theorem} [\cite{Tim}] \label{Finfty}
For every data stream of length $m$, any randomized streaming algorithm that computes $F_{\infty}$ to 
within a $(1\pm 0.2)$ factor with probability at least $2/3$ requires space $\Omega(\min \{m, n' \})$. 
\end{theorem}

We remark that, approximating $F_{\infty}$  to 
within a $(1\pm 0.2)$ factor means computing a number that is within $(1\pm 0.2)$ factor from $F_{\infty}$; however, that approximate number may not correspond to the number of occurrences of a value in the stream. 

\begin{theorem}
For every dynamic graph streaming of length $m$ for weighted graphs, any randomized streaming algorithm that, with probability at least $2/3$, 
approximates the maximum-weight $1$-matching of the graph to a $\frac{6}{5}$ factor uses space 
$\Omega(\{m, \frac{(n-1)(n-2)}{2}\})$.
\end{theorem}
\begin{proof}
Given a data stream $\SSS' = x_1, x_2, \ldots, x_m$, where each $x_i \in \{1, \ldots,   n'\}$, we define a graph stream $\SSS$ for a weighted graph $G$ on $n$ vertices, where $n$ satisfies $(n-1)(n-2)/2 < n' \le n(n-1)/2$. Let $V =\{0 \ldots, n-1\}$ be the vertex-set of $G$. We first define a bijective 
function $\chi: \{(i, j) \mid i < j \in [n]^- \}  \longrightarrow [\frac{n(n-1)}{2}]$. 
Let $\chi^{-1}$ be the inverse function of $\chi$.
Then, we can translate $\SSS'$ to a general dynamic graph streaming $\SSS$ of underlying weighted graph $G$
 by corresponding with $x_i$ the $i$-th  element $(\chi^{-1}(x_i), 1)$ of $\SSS$,  for $i\in [m]$. Observe that 
computing $F_{\infty}$ of $\SSS'$ is equivalent to computing a maximum-weight 1-matching for the graph stream  $\SSS$ of $G$. 
Let $uv$ be a maximum-weight 1-matching of $\SSS$, then $\chi(uv)$ is 
$F_{\infty}$ of $\SSS'$. By Theorem \ref{Finfty}, it follows that any randomized approximation streaming algorithm that approximates the
maximum-weight $1$-matching of $G$ to a $\frac{6}{5}$-factor with probability at least $2/3$ uses space 
$\Omega(\{m, \frac{(n-1)(n-2)}{2}\})$, thus completing the proof.
\end{proof}
\fi
\section{Insert-Only Streaming Model}
\label{sec:insert-only}
In this section, we give a streaming algorithm for {\sc p-WT-Matching}, and hence for \umatch{} as a special case, in the insert-only model. We start by defining some notations. 

Given a weighted graph $G=(V=\{0, \ldots, n-1\}, E)$ along with the weight function $wt: E(G) \longrightarrow \mathbb{R}_{\geq 0}$, and a parameter $k$, 
we define a new function $\beta: E(G) \longrightarrow \mathbb{R}_{\geq 0} \times [n]^- \times [n]^-$ as follows: for $e=uv \in E$, where $u < v$, let $\beta(e)=(wt(e), u, v)$. Observe that $\beta$ is injective.

Define a partial order relation $\prec$ on $E(G)$ as follows: for any two distinct edges $e, e' \in E(G)$, $e \prec e'$ if $\beta(e)$ is lexicographically smaller than $\beta(e')$. For a vertex $v \in V$ and an edge $e$ incident to $v$, define $\Gamma_v$ to be the sequence of edges incident to $v$, sorted in a decreasing order w.r.t.~$\prec$. We say that $e$ is the \emph{$i$-heaviest} edge w.r.t.~$v$ if $e$ is the $i$-th element in $\Gamma_v$.

\begin{comment}
 The \emph{trimmed} subgraph of $H$, denoted $\TT(H)$, is a subgraph of $H$ that contains every edge $e=uv \in H$ such that $e$ is among the first $8k$ heaviest edges incident to $u$ and the among the first 
$8k$ heaviest edges incident to $v$ in $H$. The reduced subgraph of $H$, denoted $\RR(H)$, is a subgraph of $\TT(H)$ consisting of the $k(16k-1)$ heaviest edges in $\TT(H)$ (if $\TT(H)$ has fewer than $k(16k-1)$ edges then $\RR(H) = \TT(H)$). Note that, for a fixed subgraph $H$ of $G$, $\TT(H)$ and $\RR(H)$ are well-defined and unique. 
\end{comment}

Let $f: V \longrightarrow [4k^2]^-$ be a hash function. Let $H$ be a subgraph of $G$ (possibly $G$ itself). The function $f$ partitions $V(H)$ into the set of subsets ${\cal V}= \{V_1, \ldots, V_r\}$, where each $V_i$, $i \in [r]$, consists of the vertices in $V(H)$ that have the same image under $f$. A matching $M$ in $H$ is said to be \emph{nice} w.r.t.~$f$ if no two vertices of $M$ belong to the same part $V_i$, where $i \in [r]$, in ${\cal V}$. If $f$ is clear from the context, we will simply write $M$ is nice. We define the \emph{compact} subgraph of $H$ under $f$, denoted $\CC(H,f)$, as the subgraph of $H$ consisting of each edge $uv$ in $H$ whose endpoints belong to different parts, say $u \in V_i$, $v \in V_j$, $i \neq j \in [r]$, and such that $\beta(uv)$ is maximum over all edges between $V_i$ and $V_j$; that is, $\beta(uv) = \max{\{\beta(u'v') \mid u'v' \in E(H) \wedge u' \in V_i \wedge v' \in V_j\}}$. Finally, we define the \emph{reduced compact} subgraph of $H$ under $f$, denoted $\RC(H, f)$, by (1) selecting each edge $uv \in \CC(H, f)$ such that $uv$ is among the $8k$ heaviest edges (or all edges if there are not that many edges) incident to vertices in $V_i$ and among the $8k$ heaviest edges incident to vertices in $V_j$; and then (2) letting $q=k(16k-1)$ and retaining from the selected edges in (1) the  $q$ heaviest edges (or all edges if there are not that many edges).  We have the following:

\begin{comment}
First, we create an auxiliary weighted graph $\Phi$ whose vertices are the parts $V_1, \ldots, V_r$, and whose edges are $V_iV_j$, $i \neq j \in [r]$, such that some vertex in $V_i$ is adjacent to some vertex in $V_j$; we associate with edge $V_iV_j$ the value $\beta(uv)$, where $u \in V_i$ and $v \in V_j$, and $\beta(uv)$ is the maximum of (w.r.t. $\prec$) over all such $u \in V_i$ and $v \in V_j$, we correspond the edge $uv$ with $V_iV_j$. Now consider $\RR(\Phi)$, and define $\RC(H, f)$ to be the subgraph of $H$ whose edges correspond to the edge in $\RR(\Phi)$.

Suppose that $H$ has a maximum-weight $k$-matching $M$ and $f$ 
is perfect w.r.t.~$V_M$, then $\RC (H, f)$ has a maximum-weight 
$k$-matching $M'$ such that $wt(M')=wt(M)$. 
\end{comment}

\begin{lemma}\ifshort {\rm ($\spadesuit$)} \fi \label{lemma:reduced}
The subgraph $\CC(H, f)$ has a nice $k$-matching if and only if $\RC(H, f)$ has a nice $k$-matching. Moreover, if $\CC(H, f)$ (and hence $\RC(H, f)$) has a nice $k$-matching, then the weight of a maximum-weight nice $k$-matching in $\CC(H, f)$ is equal to that in $\RC(H, f)$. 
\end{lemma}
\iflong
\begin{proof}
As discussed before, the function $f$ partitions $V(H)$ into the set $\mathcal{V}=\{V_1, \ldots, V_r\}$, where each $V_i$, $i\in [r]$, consists of 
the vertices in $V(H)$ that have the same image under $f$. 
Define the auxiliary weighted graph $\Phi$ whose vertices are the parts $V_1, \ldots, V_r$, and such that there is an edge $V_iV_j$ in $\Phi$, $i \neq j \in [r]$, if 
some vertex $u \in V_i$ is adjacent to some vertex $v \in V_j$ in $\CC(H, f)$; we associate with edge $V_iV_j$ the value $\beta(uv)$ and associate the edge $uv$ with $V_iV_j$. Obviously, there is one-to-one correspondence between the nice $k$-matchings in $\CC(H,f)$ and the $k$-matchings of $\Phi$. Let $\mathscr{H}$ be the subgraph of $\Phi$ formed by selecting each edge $V_iV_j$, $i\neq j\in [r]$, 
such that $V_iV_j$ is among the $8k$ heaviest edges (or all edges if there are not that many edges)
incident to $V_i$ and among the $8k$ heaviest edges (or all edges if there are not that many edges) 
incident to $V_j$. Let $\mathscr{H'}$ consist of the $q$ heaviest edges in $\mathscr{H}$; if $\mathscr{H}$ has at most $q$ edges, we let ${\cal H'}={\cal H}$. Since there is a one-to-one correspondence between the nice $k$-matchings in $\CC(H,f)$ and the $k$-matchings of $\Phi$, it suffices to prove the statement of the lemma with respect to matchings in $\Phi$ and ${\cal H'}$; namely, since ${\cal H'}$ is a subgraph of $\Phi$, it suffices to show that: if $\Phi$ has a maximum-weight $k$-matching $M$ then ${\cal H'}$ has a maximum-weight $k$-matching of the same weight as $M$.  

Suppose that $\Phi$ has a maximum-weight $k$-matching $M$. Choose $M$ such that the number of edges in $M$ that remain in ${\cal H}$ is maximized. We will show first that all the edges in $M$ remain in ${\cal H}$. Suppose not, then there is an edge $V_{i_0}V_{i_1} \in M$ such that $V_{i_0}V_{i_1}$ is not among the $8k$ heaviest edges incident to one of its endpoints, say $V_{i_1}$. Since $|V(M)| =2 k < 8k$, it follows that there is a heaviest edge $V_{i_1}V_{i_2}$ incident to $V_{i_1}$ such that $\beta(V_{i_1}V_{i_2}) > \beta(V_{i_0}V_{i_1})$ and $V_{i_2} \notin V_M$.  If $V_{i_1}V_{i_2} \in \mathscr{H}$, then $(M- V_{i_0}V_{i_1}) + V_{i_1}V_{i_2}$
   is a maximum-weight $k$-matching of $\Phi$ that contains more edges of $\mathscr{H}$ than $M$, contradicting our choice of $M$. It follows that 
   $V_{i_1}V_{i_2} \notin \mathscr{H}$. Then, $V_{i_1}V_{i_2}$ is not among the $8k$ 
   heaviest edges incident to $V_{i_2}$. Now apply the above argument to $V_{i_2}$ to 
   select the heaviest edge $V_{i_2}V_{i_3}$ such that $\beta(V_{i_2}V_{i_3}) > \beta(V_{i_1}V_{i_2}) > \beta(V_{i_0}V_{i_1})$ and  $V_{i_3}\notin V_M$. 
   By applying the above argument $j$ times, we obtain a sequence of $j$ vertices $V_{i_1}, V_{i_2}, \ldots, V_{i_j}$, such that 
   (1) $\{V_{i_2}, \ldots, V_{i_j}\} \cap V_M = \emptyset$; and (2) $V_{i_a} \neq V_{i_b}$ for every $a \neq b \in [j]$, which 
   is guaranteed by $\beta(V_{i_a}V_{i_{a+1}})<\beta(V_{i_{a+1}}V_{i_{a+2}})<\cdots < \beta(V_{i_{b-1}}V_{i_{b}})$ and 
   $V_{i_a}V_{i_{a+1}}$ is the heaviest edge incident to $V_{i_a}$ such that $V_{i_{a+1}} \notin V_M$. 
   Since $\Phi$ is finite, the above process must end at an edge $e$ not in $M$ and such that $\beta(e)$ exceeds $\beta(V_{i_0}V_{i_1})$, contradicting our choice of $M$. Therefore, $M \subseteq E({\cal H})$. 
   
   Now, choose a maximum-weight $k$-matching of ${\cal H}$ that maximizes the number of edges retained in ${\cal H'}$. Without loss of generality, call it $M$. We prove that the edges of $M$ are retained in ${\cal H'}$, thus proving the lemma. Suppose that this is not the case. Since each vertex in $V(M)$ has degree at most $8k$ and one of its edges must be in $M$, the number of edges in ${\cal H}$ incident to the vertices in $M$ is at most $2k(8k-1) +k = k(16k-1)=q$. It follows that there is an edge $e$ in ${\cal H'}$ whose endpoints are not in $M$ and such that $\beta(e)$ is larger than the $\beta()$ value of some edge in $M$, contradicting our choice of $M$. 
\end{proof}
\fi

\begin{lemma}\ifshort {\rm ($\spadesuit$)} \fi
\label{lemma:rc}
 Let $f: V \longrightarrow [4k^2]^-$ be a hash function, and let $H$ be a subgraph of $G$. There is an algorithm \textbf{Alg-Reduce($H$, $f$)} that computes $\RC(H, f)$ and whose time and space complexity is $\bigoh(|H|+k^2)$. 
\end{lemma}
\iflong
\begin{proof}
The algorithm \textbf{Alg-Reduce($H$, $f$)} works as follows. First, for each $v \in V(H)$, it computes $f(v)$ and uses it to partition $V(H)$ into $V_1, \ldots, V_r$, where each $V_i$, $i \in [r]$, consists of the vertices in $V(H)$ that have the same image under $f$. Clearly, the above can be done in time $\bigoh(|H|+k^2)$ (e.g., using Radix sort). Then, it partitions the edges of $H$ into groups $E_{i,j}$, $i \neq j \in [r]$, where each $E_{i,j}$ consists of all the edges in $H$ that go between $V_i$ and $V_j$. Clearly, this can be done in $\bigoh(|H|)$ time (e.g., using Radix sort on the labels of the pairs of parts containing the edges). From each group $E_{i, j}$, among all edge in $E_{i, j}$, the algorithm retains the edge $uv$ corresponding to the maximum value $\beta(uv)$, which clearly can be done in $\bigoh(|H|)$ time. Next, the algorithm groups the remaining edges into (overlapping) groups, where each group $E_i$ consists of all the edges (among the remaining edges) that are incident to the same part $V_i$, for $i \in [r]$; note that each edge appears in exactly two such groups. The algorithm now discards every edge $uv$, where $u \in V_i$, $v \in V_j$, if either $uv$ is not among the heaviest $8k$ edges in $E_i$ or is not among the $8k$ heaviest edges in $E_j$. The above can be implemented in $\bigoh(|H|)$ time by applying a linear-time ($(8k)$-th order) selection algorithm~\cite{Cormen} to each $E_i$ to select the $(8k)$-th heaviest edge in $E_i$, and then discard all edges of lighter weight from $E_i$; an edge is kept if it is kept in both groups that contain its endpoints (which can be easily done, e.g., using a Radix sort). Finally, invoking a $q$-th order selection algorithm~\cite{Cormen}, where $q=k(16k-1)$, we can retain the heaviest $q$ edges among the remaining edges; those edges form $\RC(H, f)$. The above algorithm runs in time $\bigoh(|H|+k^2)$, and its space complexity is dominated by $\bigoh(|H|+k^2)$ as well. This completes the proof. 
\end{proof}
\fi

We now present the streaming algorithm $\mathcal{A}_{Insert}$ for {\sc p-WT-Matching}. 
Let $(\SSS, k)$ be an instance of \wmatcheq, where $\SSS = (e_1,wt(e_1)), \ldots, (e_i,wt(e_i)), \ldots$ is a stream of edge-insertions for a graph $G$. 
For $i \in \mathbb{N}$, let $G_i$ be the subgraph of $G$ consisting of the first $i$ edges $e_1, \ldots, e_i$ of $\SSS$, and for $j \le i$, let $G_{j, i}$ be the subgraph of $G$ whose edges are $\{e_j, \ldots, e_i\}$; if $j > i$, we let $G_{j, i}= \emptyset$. 
%Let $q=k(16k-1)$. 
Let $f$ be a hash function chosen u.a.r.~from a 
universal set $\mathcal{H}$ of hash functions
mapping $V$ to $[4k^2]^-$. 
The algorithm $\mathcal{A}_{Insert}$, 
after processing the $i$-th element $(e_i,wt(e_i))$, 
computes two subgraphs $G_i^f, G_i^s$
defined as follows. 
For $i =0$, define $G_i^f= G_i^s=\emptyset$. 
%Initially, let $G_0^f=G_q^f=G_0^s=\emptyset$.
Suppose now that $i > 0$. Define $\hat{i}$ to be the largest multiple of $q$ that is smaller than $i$, that is, $i = \hat{i} +p$, where $0 < p \leq q$; and define $i^*$ as the largest multiple of $q$ that is smaller than $\hat{i}$ if $\hat{i} > 0$, and $0$ otherwise (i.e., $i^*=0$ if $\hat{i}=0$). The subgraph $G_i^f$ is defined only when $i$ is a multiple of $q$ (i.e., $i= j\cdot q$ where $j \geq 0$), and is defined recursively for $i = j \cdot q > 0$ as $G_{i}^{f} = \RC(G_{\hat{i}}^{f} \cup G_{i^*+1,\hat{i}})$; that is, $G_{i}^{f}$ is the reduced compact subgraph of the graph consisting of $G_{\hat{i}}^{f}$ plus the subgraph consisting of the edges encountered after $e_{i^*}$, starting from $e_{i^*+1}$ up to $e_{\hat{i}}$.  The subgraph $G_{i}^{s}$ is defined as  
$G_{i}^{s}= G_{\hat{i}}^{f} \cup G_{i^*+1, i}$; that is, $G_{i}^{s}$ consists of the previous (before $i$) reduced compact subgraph plus the subgraph consisting of the edges starting after $i^*$ up to $i$. \ifshort (Refer to Figure~\ref{fig:compact} in Section~\ref{sec:figures} for an illustration.)\fi \iflong We refer to Figure~\ref{fig:compact} for an illustration of the definitions of $G_{i}^{f}$ and $G_{i}^{s}$.\fi

\iflong
\begin{observation}
\label{obs:1}
For each $i$ that is a multiple of $q$, $G_i^f$ contains at most $q$ edges (by the definition of a reduced compact subgraph).
\end{observation} 

\begin{observation}
\label{obs:2} For each $i$, $G_i^s$ contains at most $3q$ edges. 
\end{observation} 
\fi

\begin{lemma}\ifshort {\rm ($\spadesuit$)} \fi  \label{lem15}
For each $i\ge 1$, if $G_i$ contains a maximum-weight $k$-matching, then with probability at least $1/2$, 
$G_i^s$ contains a maximum-weight $k$-matching of $G_i$. 
\end{lemma}
\iflong
\begin{proof}
Let $M=\{u_0u_1, \ldots, u_{2k-2}u_{2k-1}\}$ be a maximum-weight $k$-matching in $G_i$, and let $V_M=\{ u_0, \ldots, u_{2k-1} \}$.  Since $f$ is a hash function chosen u.a.r.~from a universal set $\mathcal{H}$ of hash functions mapping $V$ to $[4k^2]^-$, 
by Theorem \ref{lemma-hash}, with probability at least $1/2$, $f$ is perfect w.r.t.~$V_M$. 
Now, suppose that $f$ is perfect w.r.t.~$V_M$, and hence, we have 
$f(u_j) \neq f(u_l)$ for every $j \neq l \in [2k]^-$. Thus, $M$ is a nice matching (w.r.t.~$f$) in $G_i$.
%Let $G_i^h$ be the compact graph of $G_i$ under $h$. 
By the definition of $\CC(G_i,f)$, 
there is a set $M'$ of $k$ edges $M'=\{u'_0u'_1, \ldots, u'_{2k-2}u'_{2k-1}\}$ in $\CC(G_i,f)$ such that 
$\{ f(u'_{2i}), f(u'_{2i+1}) \} = \{ f(u_{2i}), f(u_{2i+1}) \}$ and $\beta(u'_{2i}u'_{2i+1}) \ge \beta(u_{2i}u_{2i+1})$
 for $i\in [k]^-$. It follows that $wt(u'_{2i}u'_{2i+1}) \ge wt(u_{2i}u_{2i+1})$ for $i\in [k]^-$. 
 Therefore, $\CC(G_i,f)$ contains a maximum-weight $k$-matching of $G_i$, namely $\{u'_0u'_1, \ldots, u'_{2k-2}u'_{2k-1}\}$; moreover, this matching is nice. 
By Lemma \ref{lemma:reduced}, $\RC(G_i,f)$ contains a maximum-weight $k$-matching 
of $\CC(G_i,f)$.

Next, we prove that $G_i^s$ contains a maximum-weight $k$-matching of $G_i$. 
If $i\le 2q$, then $G_i^s=G_{1,i}=G_i$ by definition, and hence $G_i^s$ contains a maximum-weight $k$-matching of $G_i$. 
Suppose now that $i>2q$. 
By definition, $G_{i}^{s}= G_{\hat{i}}^{f} \cup G_{i^*+1, i}$. (Recall that, by definition, 
$G_q^f$=$\emptyset$, $G_{2q}^f=\RC(G_q^f \cup G_{1,q})$, $G_{3q}^f=\RC(G_{2q}^f \cup G_{q+1,2q}),\ldots,$
$G_{\hat{i}}^f = \RC(G_{i^*}^f \cup G_{i^*-q+1, i^*})$.)
For each $j \geq 1$ that is multiple of $q$, let $\mathscr{G}_{j}$ be the graph consisting of the edges that are
in $G_{\hat{j}}^f\cup G_{j^*+1, \hat{j}}$ but are not kept in $G_{j}^f$. Consequently, 
$(\bigcup_{q\le j <\hat{i}, j \text{ is a multiple of } q} \mathscr{G}_j) \bigcup G_{\hat{i}}^f=G_{i^*}$.
By the definition of $\RC(G_i,f)$, it is easy to verify that $\RC(G_i,f)$ does not contain the edges in $\mathscr{G}_j$, for each $j\ge 1$.
It follows that $\RC(G_i,f)$ is a subgraph of 
$G_{i}^{s}$, and hence, $G_i^s$ contains a maximum-weight $k$-matching of $\RC(G_i,f)$, and hence of $\CC(G_i,f)$ by the above discussion. Since $\CC(G_i,f)$ contains a maximum-weight $k$-matching of $G_i$, 
$G_i^s$ contains a maximum-weight $k$-matching of $G_i$. It follows that,  with probability at least $1/2$, $G_i^s$ contains a maximum-weight $k$-matching of $G_i$.
\end{proof}
\fi

\iflong
\begin{center}
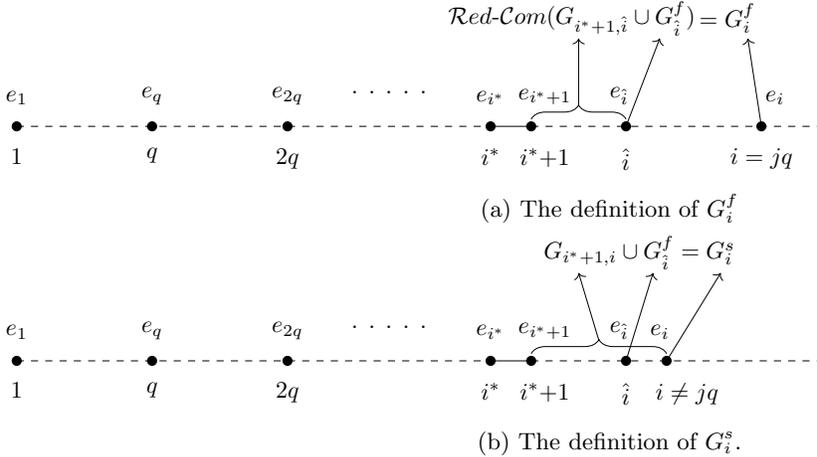
\begin{figure*}[htbp]
\begin{subfigure}{\textwidth}

%\begin{adjustbox}{\textwidth}
\scalebox{0.9}{
\begin{tikzpicture}

\tkzDefPoint(0,0){e_1}
\tkzDefPoint(2,0){e_q}
\tkzDefPoint(4,0){e_2q}

\draw[shift={(5.5, 0.5)}]   node[]{.   .   .   .   .};

\tkzDefPoint(7,0){e_i*}
\tkzDefPoint(7.6,0){e_i*+1}
\tkzDefPoint(9,0){e_i^}
\tkzDefPoint(11,0){e_i}

\tkzLabelPoint[above, yshift=0.2cm](e_1){$e_1$}
\tkzLabelPoint[above, yshift=0.2cm](e_q){$e_q$}
\tkzLabelPoint[above, yshift=0.2cm](e_2q){$e_{2q}$}

\tkzLabelPoint[above, yshift=0.2cm](e_i*){$e_{i^*}$}
\tkzLabelPoint[above, xshift=0.2cm, yshift=0.2cm](e_i*+1){$e_{i^*+1}$}
\tkzLabelPoint[above, xshift=-0.1cm, yshift=0.2cm](e_i^){$e_{\hat{i}}$}
\tkzLabelPoint[above, xshift=0.2cm, yshift=0.2cm](e_i){$e_{i}$}

\tkzLabelPoint[below, yshift=-0.2cm](e_1){1}
\tkzLabelPoint[below, yshift=-0.2cm](e_q){$q$}
\tkzLabelPoint[below, yshift=-0.2cm](e_2q){$2q$}

\tkzLabelPoint[below, yshift=-0.2cm](e_i*){$i^*$}
\tkzLabelPoint[below, xshift=0.2cm, yshift=-0.2cm](e_i*+1){$i^*$+1}
\tkzLabelPoint[below, yshift=-0.2cm](e_i^){$\hat{i}$}
\tkzLabelPoint[below, yshift=-0.2cm](e_i){$i = jq$}

\foreach \n in {e_1, e_q, e_2q, e_i*, e_i*+1, e_i^, e_i}
  \node at (\n)[circle,fill,inner sep=1.5pt]{};

    \draw[color=black, dashed] (0,0) -- (7,0);
    \draw[color=black, dashed] (7.6,0) -- (12,0);
    \draw[color=black] (7,0) -- (7.6,0);

    \draw[->](8.3,0.3) -- (8.3,1.3);
    \draw[shift={(8.2,1.6)}]   node[]{$\RC(G_{i^*+1,\hat{i}} \cup G_{\hat{i}}^f)$};

    \draw[->](9,0) -- (9.5,1.3);

    \draw[->](11,0) -- (10.8,1.3);
    \draw[shift={(10.5,1.6)}]   node[]{$ = G_{i}^f$};

%\draw (8.3,0) ellipse (0.9cm and 0.3cm);

\draw [decorate,decoration={brace,amplitude=6pt,mirror,raise=3pt},yshift=0pt]
(9,0) -- (7.6,0) node [black,midway,xshift=1cm] {};

%\AsymCloud{(8.75,1.7)}{}{1.1}

\end{tikzpicture}

}
%\end{adjustbox}
\caption{The definition of $G_i^f$}

%\vspace{-0.4cm}

\end{subfigure}
\begin{subfigure}{\textwidth}
%\begin{adjustbox}{\textwidth}
\scalebox{0.9}{
\begin{tikzpicture}

\tkzDefPoint(0,0){e_1}
\tkzDefPoint(2,0){e_q}
\tkzDefPoint(4,0){e_2q}

\draw[shift={(5.5, 0.5)}]   node[]{.   .   .   .   .};

\tkzDefPoint(7,0){e_i*}
\tkzDefPoint(7.6,0){e_i*+1}
\tkzDefPoint(9,0){e_i^}
\tkzDefPoint(9.6,0){e_i}

\tkzLabelPoint[above, yshift=0.2cm](e_1){$e_1$}
\tkzLabelPoint[above, yshift=0.2cm](e_q){$e_q$}
\tkzLabelPoint[above, yshift=0.2cm](e_2q){$e_{2q}$}

\tkzLabelPoint[above, yshift=0.2cm](e_i*){$e_{i^*}$}
\tkzLabelPoint[above, xshift=0.2cm, yshift=0.2cm](e_i*+1){$e_{i^*+1}$}
\tkzLabelPoint[above, xshift=-0.1cm, yshift=0.2cm](e_i^){$e_{\hat{i}}$}
\tkzLabelPoint[above, xshift=-0.1cm, yshift=0.2cm](e_i){$e_{i}$}

\tkzLabelPoint[below, yshift=-0.2cm](e_1){1}
\tkzLabelPoint[below, yshift=-0.2cm](e_q){$q$}
\tkzLabelPoint[below, yshift=-0.2cm](e_2q){$2q$}

\tkzLabelPoint[below, yshift=-0.2cm](e_i*){$i^*$}
\tkzLabelPoint[below, xshift=0.2cm, yshift=-0.2cm](e_i*+1){$i^*$+1}
\tkzLabelPoint[below, yshift=-0.2cm](e_i^){$\hat{i}$}
\tkzLabelPoint[below, xshift=0.3cm, yshift=-0.2cm](e_i){$i \neq jq$}

\foreach \n in {e_1, e_q, e_2q, e_i*, e_i*+1, e_i^, e_i}
  \node at (\n)[circle,fill,inner sep=1.5pt]{};

    \draw[color=black, dashed] (0,0) -- (7,0);
    \draw[color=black, dashed] (7.6,0) -- (12,0);
    \draw[color=black] (7,0) -- (7.6,0);

    \draw[->](8.6,0.3) -- (8.3,1.3);
    \draw[shift={(9.2,1.6)}]   node[]{$G_{i^*+1,i} \cup G_{\hat{i}}^f = G_{i}^s$};

    \draw[->](9,0) -- (9.4,1.3);

    \draw[->](9.6,0) -- (10.4,1.3);

%\draw (8.3,0) ellipse (0.9cm and 0.3cm);

\draw [decorate,decoration={brace,amplitude=6pt,mirror,raise=3pt},yshift=0pt]
(9.6,0) -- (7.6,0) node [black,midway,xshift=1cm] {};

%\AsymCloud{(9.1,1.7)}{}{0.6}

\end{tikzpicture}
}
%\end{adjustbox}
\caption{The definition of $G_i^s$.}

%\vspace{-0.4cm}

\end{subfigure}

\caption{Illustration of the definitions of $G_i^f$ and $G_i^s$.}
\label{fig:compact}

\end{figure*}
\end{center}
\fi

The algorithm ${\cal A}_{Insert}$, when queried at the end of the stream, either returns a maximum-weight $k$-matching of $G$ or the empty set. To do so, at every instant $i$, it will maintain a subgraph $G_{i}^{s}$ that will contain the edges of the desired matching, from which this matching can be extracted. To maintain $G_{i}^{s}$, the algorithm keeps track of the subgraphs $G_{i-1}^s$, $G_{\hat{i}}^{f}$, the edges $e_{i^*+1}, \ldots, e_i$, and will use them in the computation of the subgraph $G_{i}^s$ as follows. If $i$ is not a multiple of $q$, then $G_{i}^s=G_{i-1}^s +e_i$, and the algorithm simply computes $G_{i}^s$ as such. Otherwise (i.e., $i$ is a multiple of $q$), $G_{i}^{s} = G_{\hat{i}}^{f} \cup G_{i^*+1, i}$, and the algorithm uses $G_{\hat{i}}^{f}$ and $G_{i^*+1, i} = \{e_{i^*+1}, \ldots, e_i\}$ to compute and return $G_{i}^s$; however, in this case (i.e., $i$ is a multiple of $q$), the algorithm will additionally need to have $G_{i}^f$ already computed, in preparation for the potential computations of subsequent $G_{j}^s$, for $j \ge i$. By Lemma~\ref{lemma:rc}, the subgraph $G_{i}^{f}$ can be computed by invoking the \textbf{Alg-Reduce} in Lemma~\ref{lemma:rc} on $G_{\hat{i}}^{f} \cup G_{i^*+1, \hat{i}}$, which runs in time $\bigoh(q)$. Note that both $G_{\hat{i}}^{f}$ and $G_{i^*+1, \hat{i}}$ are available to ${\cal A'}$ at each of the steps $\hat{i}+1, \ldots, i$. Therefore, the algorithm will stagger the $\bigoh(q)$ many operations needed for the computation of $G_{i}^{f}$ uniformly (roughly equally) over each of the steps $\hat{i}+1, \ldots, i$, yielding an $\bigoh(1)$ operations per step. Note that all the operations in~\textbf{Alg-Reduce} can be explicitly listed, and hence, splitting them over an interval of $q$ steps is easily achievable. \ifshort Combining the above discussions and lemmas, we conclude with:
\fi
\iflong Suppose that \textbf{Alg-Reduce} is split into $q$ operations $\Lambda_1, \ldots, \Lambda_q$ such that each 
operation takes time $\bigoh(1)$.  The algorithm $\cal A'$ is given in Figure \ref{fig4}:

\begin{algorithm}[htbp]
\caption{The streaming algorithm $\mathcal{A}_{Insert}$ in the insert-only streaming model}\label{fig4}
%\textbf{Input:} $n=|V(G)|$ and a parameter $k \in \nat$  

\begin{algorithmic}[1]

\vspace*{3mm}
\Algphase{{\bf $\mathcal{A}_{Insert}$-Preprocessing}: The preprocessing algorithm}
\Require $n=|V(G)|$ and a parameter $k \in \nat$

\State let $f \in_{u.a.r.} \mathcal{H}$, where $\mathcal{H}$ is a set of universal hash functions from $V$ to $[4k^2]^-$
 
\vspace*{3mm}
\Algphase{{\bf $\mathcal{A}_{Insert}$-Update}: The update algorithm}
\Require The $i$-th element $(e_i=uv, wt(e)) \in \SSS$
\setcounter{ALG@line}{0}

       \State let $j=i\mod q$
        \If { $j$ is $0$ }
            \State execute $\Lambda_q$
            \State $G_i^s=G_{\hat{i}}^f \cup G_{i^*+1,i}$
        \Else
            \State execute $\Lambda_j$
            \State $G_i^s = G_{i-1}^s \cup e_i$
        \EndIf
      
\vspace*{3mm}
\Algphase{{\bf $\mathcal{A}_{Insert}$-Query}: An algorithm to answer query after the $i$-th update}

\setcounter{ALG@line}{0}

\State return a maximum-weight $k$-matching in $G_i^s$
if any; otherwise, return $\emptyset$
       
\end{algorithmic}
\end{algorithm}

\begin{lemma}\ifshort {\rm ($\spadesuit$)} \fi  \label{lem16}
The algorithm $\mathcal{A}_{Insert}$ runs in space $\bigoh(k^2)$ and has update time $\bigoh(1)$. 
\end{lemma}

\begin{proof}
{\bf $\mathcal{A}_{Insert}$-Preprocessing} takes $\bigoh(1)$ space to store $f$. 
The space needed for {\bf $\mathcal{A}_{Insert}$-Update}
is dominated by that needed for storing $G_i^s, G_i^f$, $G_{i^*+1,i}, G_{i^*+1,\hat{i}}$, 
and the space needed to execute~\textbf{Alg-Reduce}. 
By the definition of the subgraphs $G_i^s$, $G_{i}^{f}$, 
$G_{i^*+1,i}$, $G_{i^*+1,\hat{i}}$, each has size $\bigoh(q)$, and hence can be stored using  $\bigoh(q)=\bigoh(k^2)$ space. By Lemma~\ref{lemma:rc}, 
\textbf{Alg-Reduce} runs in space $\bigoh(k^2)$. Hence, the overall space complexity of ${\cal A'}$ is $\bigoh(k^2)$. 
{\bf $\mathcal{A}_{Insert}$-Query} takes space $\bigoh(q)=\bigoh(k^2)$ \cite{gabow1,gabow2}, since $G_i^s$ contains at most $\bigoh(q)$ edges.
Altogether, {\bf $\mathcal{A}_{Insert}$} runs in space $\bigoh(k^2)$. 

For the update time, as discussed above, we can take the operations performed by 
Algorithm~\textbf{Alg-Reduce} to compute $G_{i}^{f}$, during any interval of $q$ steps, and stagger them uniformly over the $q$ steps of the interval.  
Since~\textbf{Alg-Reduce} performs $\bigoh(q)$ many operations to compute $G_{i}^{f}$ by Lemma~\ref{lemma:rc}, this yields an $\bigoh(1)$ operations per step. It follows that the update time of $\mathcal{A}_{Insert}$ is $\bigoh(1)$, thus completing the proof. 
\end{proof}

Without loss of generality, and for convenience, we will assume that the algorithm is 
queried at the end of the stream $\SSS$, even though the query could take place after any arbitrary operation $i$. \fi
\begin{theorem}\ifshort {\rm ($\spadesuit$)} \fi
\label{thm:insert}
Let $0< \delta <1$ be a parameter. In the insert-only streaming model, there is an algorithm
for {\sc p-WT-Matching} such that, on input $(\SSS, k)$, the algorithm outputs a matching $M'$ satisfying that (1) if $G$ contains a $k$-matching then, 
with probability at least $1-\delta$, $M'$ is a maximum-weight $k$-matching of $G$; and
(2) if $G$ does not contain a $k$-matching then $M'=\emptyset$. The algorithm runs in 
$\bigoh(k^2\log \frac{1}{\delta})$ space and has $\bigoh(\log \frac{1}{\delta})$ update time. In particular, for any constant $\delta$, the algorithm runs in space $\bigoh(k^2)$ and has $\bigoh(1)$ update time.
\end{theorem}
\iflong
\begin{proof}
Run $\lceil \log \frac{1}{\delta} \rceil$-many copies of algorithm $\mathcal{A}_{Insert}$  in parallel (i.e., using dove-tailing). Then, 
by the end of the stream, there are $\lceil \log \frac{1}{\delta} \rceil$ copies of $G_m^s$, where $m$ is the length of the stream. Let $G'$ be the union of 
all the $G_m^s$'s produced by the runs of $\mathcal{A}_{Insert}$. If $G'$ has a $k$-matching, let $M'$ be a maximum-weight $k$-matching of $G'$; otherwise, let $M'=\emptyset$.

By Lemma \ref{lem15}, if $G_m$, i.e., $G$, contains a maximum-weight $k$-matching, with probability at least $1/2$,
one copy of $G_m^s$ contains a maximum-weight $k$-matching 
of $G$. Hence, with probability at least 
$1-(1/2)^{\lceil \log \frac{1}{\delta} \rceil} \ge 1-\delta$, $G'$ contains a 
maximum-weight $k$-matching of $G$. It follows that 
 if $G$ contains a maximum-weight $k$-matching $M$ then, 
with probability at least $1-\delta$, $G'$ contains a 
maximum-weight $k$-matching of the same weight as $M$ and hence $M'$ is a maximum-weight $k$-matching of $G$.

Observe that the graph $G'$ is a subgraph of $G$. Therefore, 
statement (2) in the theorem clearly holds true.

By Lemma \ref{lem16}, the above algorithm runs in space $\bigoh(k^2\log \frac{1}{\delta})$ and has 
update time $\bigoh(\log \frac{1}{\delta})$, thus completing the proof. 
\end{proof}
\fi

\section{Concluding Remarks}
\label{sec:conclusion}
In this paper, we presented streaming algorithms for the $k$-matching problem in both the dynamic and insert-only streaming models. Our results improve previous works and achieve optimal space and update-time complexity. Our result for the weighted $k$-matching problem was achieved using a newly-developed  structural result that is of independent interest. 

An obvious open question that ensues from our work is whether the dependency on the number of distinct weights $W$ in our result, and the result in~\cite{Chitnis2016} as well, for weighted $k$-matching in the dynamic streaming model can be lifted. More specifically, does there exist a dynamic streaming algorithm for {\sc p-WT-Matching} whose space complexity is $\tilde{O}(k^2)$ and update time is $\tilde{O}(1)$? We leave this as an (important) open question for future research.

\newpage
\bibliography{ref}

\ifshort
\newpage 
\section{Appendix: List of Figures}
\label{sec:figures}
\begin{figure*}[ht]
\begin{center}
%\begin{adjustbox}{\textwidth}
\scalebox{0.9}{
\begin{tikzpicture}
\draw[shift={(0,3)}]   node[above]{$U$};
\draw (0,0) ellipse (1cm and 3cm);
\draw[shift={(4,3)}]   node[above]{$U$};
\draw (4,0) ellipse (1cm and 3cm);

     \draw[shift={(2,0)}]   node[above]{$f \in_{u.a.r.} {\mathcal H}$};
     \draw[->](1.1,0) -- (2.9,0);

\draw[shift={(4,2.5)}]   node[]{$U_0$};
\draw[shift={(4,1.6)}]   node[rotate=90]{. . .};
\draw[shift={(4,0.55)}]   node[]{$U_{j-1}$};
\draw[shift={(4,-0.15)}]   node[]{$U_{j}$};
\draw[shift={(4,-0.9)}]   node[]{$U_{j+1}$};
\draw[shift={(4,-1.85)}]   node[rotate=90]{. . .};
\draw[shift={(4,-2.6)}]   node[]{$U_{d_1-1}$};

\tkzDefPoint(3.35, 2.3){A}\tkzDefPoint(4, 2.15){B}\tkzDefPoint(4.65, 2.3){C}
\tkzCircumCenter(A,B,C)\tkzGetPoint{O}
\tkzDrawArc(O,A)(C)

\tkzDefPoint(3.08, 1.2){A}\tkzDefPoint(4, 1.05){B}\tkzDefPoint(4.92, 1.2){C}
\tkzCircumCenter(A,B,C)\tkzGetPoint{O}
\tkzDrawArc(O,A)(C)   

\tkzDefPoint(3, 0.4){A}\tkzDefPoint(4, 0.25){B}\tkzDefPoint(5, 0.4){C}
\tkzCircumCenter(A,B,C)\tkzGetPoint{O}
\tkzDrawArc(O,A)(C)  

\tkzDefPoint(3, -0.4){A}\tkzDefPoint(4, -0.55){B}\tkzDefPoint(5, -0.4){C}
\tkzCircumCenter(A,B,C)\tkzGetPoint{O}
\tkzDrawArc(O,A)(C) 

\tkzDefPoint(3.08, -1.2){A}\tkzDefPoint(4, -1.35){B}\tkzDefPoint(4.92, -1.2){C}
\tkzCircumCenter(A,B,C)\tkzGetPoint{O}
\tkzDrawArc(O,A)(C)  

\tkzDefPoint(3.35, -2.25){A}\tkzDefPoint(4, -2.4){B}\tkzDefPoint(4.65, -2.25){C}
\tkzCircumCenter(A,B,C)\tkzGetPoint{O}
\tkzDrawArc(O,A)(C)

    \draw[dotted, color=black, line width=1] (10,4.4) -- (10,-4.4);

    \draw[line width=1, color=green] (10,4) -- (10,2);
    \draw[line width=1, color=red] (10,2) -- (10,-2);
    \draw[line width=1, color=blue] (10,-2) -- (10,-4);
    
\draw[-](9.8,4) -- (10.2,4);
\draw[-](9.8,2) -- (10.2,2);
    \draw[-](9.9,1.2) -- (10.1,1.2);
    \draw[-](9.9,0.4) -- (10.1,0.4);
    \draw[-](9.9,-1.2) -- (10.1,-1.2);
\draw[-](9.8,-2) -- (10.2,-2);
\draw[-](9.8,-4) -- (10.2,-4);

\tkzDefPoint(4.8,0){x}
\tkzLabelPoint[below](x){$x$}
\foreach \n in {x}
  \node at (\n)[circle,fill,inner sep=1.5pt]{};

    \draw[shift={(8.2, 1)}]   node[above, rotate=18]{using $h_1 \in F_j$};
    \draw[->](4.8,0) -- (9.8,1.6);
    \draw[shift={(10,1.6)}]   node[right]{$I_{j\cdot d_2}$};
    
    \draw[->](4.8,0) -- (9.8,0.75);
    \draw[shift={(10,0.75)}]   node[right]{$I_{j\cdot d_2+1}$};
    
    \draw[shift={(8.2, -1.18)}]   node[above, rotate=-18]{using $h_{d_2} \in F_j$};
    \draw[->](4.8,0) -- (9.8,-1.6);
    \draw[shift={(10,-1.6)}]   node[right]{$I_{(j+1)\cdot d_2-1}$};

\draw [decorate,decoration={brace,amplitude=10pt,mirror,raise=4pt},yshift=0pt]
(11.05,2) -- (11.05,4) node [black,midway,xshift=1cm] {\footnotesize
$I'_{j-1}$};

\draw [decorate,decoration={brace,amplitude=10pt,mirror,raise=4pt},yshift=0pt]
(11.05,-2) -- (11.05,2) node [black,midway,xshift=0.8cm] {\footnotesize
$I'_{j}$};

\draw [decorate,decoration={brace,amplitude=10pt,mirror,raise=4pt},yshift=0pt]
(11.05,-4) -- (11.05,-2) node [black,midway,xshift=1cm] {\footnotesize
$I'_{j+1}$};

\draw [decorate,decoration={brace,amplitude=8pt,mirror,raise=3pt},yshift=0pt]
(6,-0.5) -- (6,0.5) node [black,midway,xshift=1cm] {\footnotesize
${\mathcal G}(x)$};

\end{tikzpicture}
}
%\end{adjustbox}
\end{center}
%\vspace{-0.4cm}
\caption{Illustration for Algorithms~\ref{good-hash} and ~\ref{alg-mapping}.}
\label{fig:process}

\end{figure*}

\begin{center}
\begin{figure*}[ht]
\begin{subfigure}{\textwidth}

%\begin{adjustbox}{\textwidth}
\scalebox{0.9}{
\begin{tikzpicture}

\tkzDefPoint(0,0){e_1}
\tkzDefPoint(2,0){e_q}
\tkzDefPoint(4,0){e_2q}

\draw[shift={(5.5, 0.5)}]   node[]{.   .   .   .   .};

\tkzDefPoint(7,0){e_i*}
\tkzDefPoint(7.6,0){e_i*+1}
\tkzDefPoint(9,0){e_i^}
\tkzDefPoint(11,0){e_i}

\tkzLabelPoint[above, yshift=0.2cm](e_1){$e_1$}
\tkzLabelPoint[above, yshift=0.2cm](e_q){$e_q$}
\tkzLabelPoint[above, yshift=0.2cm](e_2q){$e_{2q}$}

\tkzLabelPoint[above, yshift=0.2cm](e_i*){$e_{i^*}$}
\tkzLabelPoint[above, xshift=0.2cm, yshift=0.2cm](e_i*+1){$e_{i^*+1}$}
\tkzLabelPoint[above, xshift=-0.1cm, yshift=0.2cm](e_i^){$e_{\hat{i}}$}
\tkzLabelPoint[above, xshift=0.2cm, yshift=0.2cm](e_i){$e_{i}$}

\tkzLabelPoint[below, yshift=-0.2cm](e_1){1}
\tkzLabelPoint[below, yshift=-0.2cm](e_q){$q$}
\tkzLabelPoint[below, yshift=-0.2cm](e_2q){$2q$}

\tkzLabelPoint[below, yshift=-0.2cm](e_i*){$i^*$}
\tkzLabelPoint[below, xshift=0.2cm, yshift=-0.2cm](e_i*+1){$i^*$+1}
\tkzLabelPoint[below, yshift=-0.2cm](e_i^){$\hat{i}$}
\tkzLabelPoint[below, yshift=-0.2cm](e_i){$i = jq$}

\foreach \n in {e_1, e_q, e_2q, e_i*, e_i*+1, e_i^, e_i}
  \node at (\n)[circle,fill,inner sep=1.5pt]{};

    \draw[color=black, dashed] (0,0) -- (7,0);
    \draw[color=black, dashed] (7.6,0) -- (12,0);
    \draw[color=black] (7,0) -- (7.6,0);

    \draw[->](8.3,0.3) -- (8.3,1.3);
    \draw[shift={(8.2,1.6)}]   node[]{$\RC(G_{i^*+1,\hat{i}} \cup G_{\hat{i}}^f)$};

    \draw[->](9,0) -- (9.5,1.3);

    \draw[->](11,0) -- (10.8,1.3);
    \draw[shift={(10.5,1.6)}]   node[]{$ = G_{i}^f$};

%\draw (8.3,0) ellipse (0.9cm and 0.3cm);

\draw [decorate,decoration={brace,amplitude=6pt,mirror,raise=3pt},yshift=0pt]
(9,0) -- (7.6,0) node [black,midway,xshift=1cm] {};

%\AsymCloud{(8.75,1.7)}{}{1.1}

\end{tikzpicture}

}
%\end{adjustbox}
\caption{The definition of $G_i^f$}

%\vspace{-0.4cm}

\end{subfigure}

\begin{subfigure}{\textwidth}
%\begin{adjustbox}{\textwidth}
\scalebox{0.9}{
\begin{tikzpicture}

\tkzDefPoint(0,0){e_1}
\tkzDefPoint(2,0){e_q}
\tkzDefPoint(4,0){e_2q}

\draw[shift={(5.5, 0.5)}]   node[]{.   .   .   .   .};

\tkzDefPoint(7,0){e_i*}
\tkzDefPoint(7.6,0){e_i*+1}
\tkzDefPoint(9,0){e_i^}
\tkzDefPoint(9.6,0){e_i}

\tkzLabelPoint[above, yshift=0.2cm](e_1){$e_1$}
\tkzLabelPoint[above, yshift=0.2cm](e_q){$e_q$}
\tkzLabelPoint[above, yshift=0.2cm](e_2q){$e_{2q}$}

\tkzLabelPoint[above, yshift=0.2cm](e_i*){$e_{i^*}$}
\tkzLabelPoint[above, xshift=0.2cm, yshift=0.2cm](e_i*+1){$e_{i^*+1}$}
\tkzLabelPoint[above, xshift=-0.1cm, yshift=0.2cm](e_i^){$e_{\hat{i}}$}
\tkzLabelPoint[above, xshift=-0.1cm, yshift=0.2cm](e_i){$e_{i}$}

\tkzLabelPoint[below, yshift=-0.2cm](e_1){1}
\tkzLabelPoint[below, yshift=-0.2cm](e_q){$q$}
\tkzLabelPoint[below, yshift=-0.2cm](e_2q){$2q$}

\tkzLabelPoint[below, yshift=-0.2cm](e_i*){$i^*$}
\tkzLabelPoint[below, xshift=0.2cm, yshift=-0.2cm](e_i*+1){$i^*$+1}
\tkzLabelPoint[below, yshift=-0.2cm](e_i^){$\hat{i}$}
\tkzLabelPoint[below, xshift=0.3cm, yshift=-0.2cm](e_i){$i \neq jq$}

\foreach \n in {e_1, e_q, e_2q, e_i*, e_i*+1, e_i^, e_i}
  \node at (\n)[circle,fill,inner sep=1.5pt]{};

    \draw[color=black, dashed] (0,0) -- (7,0);
    \draw[color=black, dashed] (7.6,0) -- (12,0);
    \draw[color=black] (7,0) -- (7.6,0);

    \draw[->](8.6,0.3) -- (8.3,1.3);
    \draw[shift={(9.2,1.6)}]   node[]{$G_{i^*+1,i} \cup G_{\hat{i}}^f = G_{i}^s$};

    \draw[->](9,0) -- (9.4,1.3);

    \draw[->](9.6,0) -- (10.4,1.3);

%\draw (8.3,0) ellipse (0.9cm and 0.3cm);

\draw [decorate,decoration={brace,amplitude=6pt,mirror,raise=3pt},yshift=0pt]
(9.6,0) -- (7.6,0) node [black,midway,xshift=1cm] {};

%\AsymCloud{(9.1,1.7)}{}{0.6}

\end{tikzpicture}
}
%\end{adjustbox}
\caption{The definition of $G_i^s$.}

%\vspace{-0.4cm}

\end{subfigure}

\caption{Illustration of the definitions of $G_i^f$ and $G_i^s$.}
\label{fig:compact}

\end{figure*}
\end{center}
\fi

\end{document}